\documentclass[10pt,twocolumn,twoside]{IEEEtran}
\IEEEoverridecommandlockouts

\usepackage{graphicx,epstopdf}
\usepackage{amsmath,amssymb,amsthm} 
\usepackage{multirow}
\usepackage{subfigure}
\usepackage{cite,color}
\usepackage{caption}
\usepackage{environ}         
\usepackage{etoolbox}        
\usepackage{pgfplots}
\usepgfplotslibrary{statistics}
\pgfplotsset{width=10.5cm, compat=1.5}
\usepackage{bbm}
\usepackage{xcolor}
\usepackage{tikz}
\usepackage{mathtools}
\usepackage{setspace}
\usepackage{algpseudocode}
\usepackage{algorithm}
\usepackage{booktabs}
\usepackage{array}
\usepackage{tabularx}
\usepackage{calc}
\usepackage{makecell}
\usepackage{float}
\usepackage{slashbox}
\usepackage{enumitem}
\usepackage{scrextend}
\usepackage{siunitx}
\usepackage[colorlinks,linkcolor=blue,citecolor=red,urlcolor=blue,bookmarks=false,hypertexnames=true]{hyperref} 
\makeatletter
\newcommand*{\rom}[1]{\expandafter\@slowromancap\romannumeral #1@}
\makeatother
\allowdisplaybreaks

\newtheorem{lemma}{Lemma}

\newtheorem{definition}{Definition}
\newtheorem{proposition}{Proposition}
\newtheorem{remark}{Remark}
\usepackage[utf8]{inputenc}

\newcommand{\ac}{algebraic connectivity}
\newcommand{\Ac}{Algebraic connectivity}
\newcommand{\oa}{outer-approximation}

\title{Optimal Robust Network Design: Formulations and Algorithms for Maximizing Algebraic Connectivity}
\author{Neelkamal Somisetty\IEEEauthorrefmark{3}, Harsha Nagarajan\IEEEauthorrefmark{2},  Swaroop Darbha\IEEEauthorrefmark{3} 
\thanks{\IEEEauthorrefmark{2} Applied Mathematics and Plasma Physics (T-5), Los Alamos National Laboratory, Los Alamos, NM, USA. Email:  \href{mailto:harsha@lanl.gov}{harsha@lanl.gov}}
\thanks{\IEEEauthorrefmark{3} Department of Mechanical Engineering, Texas A \& M University,
College Station, TX, USA. Email:  \href{mailto:neelkamal.sept18@tamu.edu}{neelkamal.sept18@tamu.edu}, \href{mailto:dswaroop@tamu.edu}{dswaroop@tamu.edu}}
}

\begin{document}

\maketitle
\thispagestyle{empty}
\pagestyle{empty}
\begin{abstract}
This paper focuses on designing edge-weighted networks, whose robustness is characterized by maximizing {\ac}, or \textcolor{black}{the second smallest eigenvalue of the Laplacian matrix.} This problem is motivated by cooperative vehicle localization, where accurately estimating relative position measurements and establishing communication links are essential. We also examine an associated problem where every robot is limited by payload, budget, and communication to pick no more than a specified number of relative position measurements. The basic underlying formulation for these problems is nonlinear and is known to be NP-hard. Our approach formulates this problem as a Mixed Integer Semi-Definite Program (MISDP), later reformulated into a Mixed Integer Linear Program (MILP) for obtaining optimal solutions using cutting plane algorithms. We introduce a novel upper-bounding algorithm based on principal minor characterization of positive semi-definite matrices and discuss a degree-constrained lower bounding formulation inspired by robust network structures. \textcolor{black}{In addition, we propose a maximum cost heuristic with low computational complexity to identify high-quality feasible solutions for instances involving up to one hundred nodes.} We show extensive computational results corroborating our proposed methods.
\end{abstract}


\section{Introduction}
\label{Sec:intro}

\IEEEPARstart{S}{ynthesis} of networks with robust connectivity/rigidity is crucial for emerging engineering applications. {\Ac} of a network, denoted as the second smallest eigenvalue of the Laplacian matrix, serves as a robustness metric, gaining significant interest in both graph theory \cite{shakeri2016maximizing, jamakovic2007relationship} and engineering perspectives \cite{varshney2010distributed}. 
In synthesizing ad-hoc UAV rigid formations, maximizing {\ac} subject to an area coverage constraint is a key problem \cite{nagarajan2015synthesizing}. {\Ac} serves as a criterion to sparsify networks in simultaneous localization and mapping (SLAM). Here, network edges correspond to memory requirements for storing observations and computational expenses for state estimation algorithms, that grow unbounded during long-term navigation \cite{doherty2022spectral, boumal2014cramer}. In \cite{wei2013algebraic}, weighted {\ac} is used to analyze the robustness of air transportation networks. \textcolor{black}{In multi-agent networks, networks with higher {\ac} are preferred due to their correlation with faster convergence speeds in consensus algorithms \cite{ogiwara2015maximizing}.} In optical communication satellite networks, {\ac} serves as a robustness measure for degree-constrained spanning trees \cite{liu2019dynamic, zheng2017weighted}. In \cite{cheung2021improving}, the authors propose to enhance the connectivity of a compromised digital logistics network, subject to budget constraints, by maximizing its {\ac}. {\Ac} also characterizes properties in multi-layer networks with arbitrary interconnections \cite{tavasoli2020maximizing}. In distributed sensor networks \cite{griparic2021algebraic} and multi-agent systems \cite{sharma2014observability,6112245}, network connectivity is expressed by {\ac}. \textcolor{black}{In \cite{ding2020connectivity}, authors use {\ac} as a measure to maintain well-connected mobile networks.} 

This paper addresses a simplified version of the robust network synthesis problem, the resolution of which remains open. The problem aims to find a sub-network with at most $q$ (a constant) edges from a complete weighted network, maximizing the {\ac} of the weighted Laplacian of the sub-network. However, this problem is NP-hard \cite{mosk2008maximum}. \textcolor{black}{In \cite{nagarajan2012algorithms}, the authors proposed an iterative primal-dual algorithm to identify a spanning tree network with maximum {\ac} from a complete weighted graph. However, it was computationally intractable for graphs with more than nine nodes.} Several neighborhood search heuristic methods have been proposed for obtaining sub-optimal solutions without quality guarantees \cite{trimble2019connectivity, ghosh_growing_2006,nagarajan2014heuristics}. The fragment and selection-merging (FSM) heuristic algorithm \cite{son2010building} iteratively merges network fragments until a spanning tree is formed. Evaluating the quality of heuristic solutions (lower bounds) often involves the relaxation of binary variables, which typically results in weak upper bounds \cite{wei2013algebraic,zheng2017weighted}. An earlier version of our work \cite{nagarajan2015maximizing} introduced an upper-bounding formulation based on Fiedler vectors, but its quality varied based on feasible solutions and the number of Fiedler vectors used. While theoretical upper bounds exist for optimal {\ac} on unweighted networks \cite{ghosh2006upper}, rigorous methods for obtaining tight upper bounds for weighted networks remains challenging. 

Distinguishing itself from \cite{nagarajan2015maximizing}, this paper's key contributions are: $(i)$ we derive and show that {\ac} serves as a robustness measure for the problem of cooperative vehicle localization with noisy measurements. We also explore a related problem concerning networks with degree constraints (payload/budget/communication), which is of independent interest, $(ii)$ we introduce a \textit{novel} upper bounding formulation and algorithm for maximizing {\ac} of weighted networks, formulated as a mixed-integer semi-definite program (MISDP). This approach is based on the hierarchy of principal minor characterization of a positive semi-definite (PSD) matrix. We also derive relatively simpler mixed integer linear and second-order conic programs, which relax the MISDP and provide tight upper bounds, $(iii)$  we propose a degree-constrained lower bounding formulation (DCLBF), which mimics the structures of robust networks, thereby considerably reducing the search space of feasible solutions, $(iv)$ \textcolor{black}{lastly, we propose a ranking-based ``maximum cost heuristic'' to efficiently obtain high-quality solutions for the DCLBF, improving scalability for networks with \textit{up to 100 nodes} and outperforming the FSM algorithm \cite{son2010building}.}

The article is structured as follows: Section \rom{2} introduces the problem of maximizing {\ac} for a vehicle localization application. Section \rom{3} presents a mathematical formulation as an MISDP with connectivity constraints. Sections \rom{4} and \rom{5} discuss an upper-bounding formulation based on principal minor characterization and a degree-constrained lower bounding formulation (DCLBF), respectively, to efficiently solve the MISDP. Section \rom{6} proposes a heuristic for quickly finding high-quality feasible solutions. Finally, computational results and concluding remarks are provided in Sections \rom{7} and \rom{8}, respectively.

\section{Maximization of {\Ac}}
\label{Sec:obj}

In this section, we emphasize the significance of choosing {\ac} as the maximization objective when searching for a sub-network within a weighted network. We illustrate this importance using the cooperative vehicle localization application, drawing motivation from \cite{1067998}. 

\textbf{Notation.} In the following sections, we use lower and upper case to represent scalars (vector/matrix elements). Bold font with lower and upper cases to represent vectors and matrices, respectively. The tensor product of two vectors $\mathbf{v_1, v_2}$ in the same vector space is denoted by $\mathbf{v_1 \otimes v_2}$, and their dot product by $\mathbf{v_1 \cdot v_2}$. For any vector $\mathbf{v} \in \mathbb{R}^n$, $\lVert \mathbf{v}\rVert: \mathbb{R}^n \rightarrow \mathbb{R}$ defines the 2-norm, given by $\sqrt{\mathbf{v \cdot v}}$. For any matrix $\mathbf{M} \in \mathbb{R}^{n\times n}$, $\mathbf{M}'$ denotes the transpose of $\mathbf{M}$ and $\lVert \mathbf{M}\rVert: \mathbb{R}^{n \times n} \rightarrow \mathbb{R}$ defines the spectral norm, given by the largest singular value of $\mathbf{M}$. Given two square symmetric matrices $\mathbf{A}$ and $\mathbf{B}$, $\mathbf{A} \succeq \mathbf{B}$ implies $\mathbf{A-B} \succeq 0$, i.e., $\mathbf{A-B}$ is a positive semi-definite (PSD) matrix. Let $\mathbf{e_i}$ denote the $i^{th}$ column of the identity matrix $\mathbf{I}_n$ of size $n \times n$. \textcolor{black}{Let $\mathbf{0}_{n \times n}$ be a zero matrix of size $n \times n$ and $\mathbf{1}$ be an $n$-dimensional vector of ones. For any non-empty set $S$, the notation $|S|$ represents the cardinality of the set, while $\emptyset$ refers to an empty set.}

\subsection{Cooperative vehicle localization with noisy measurements}
\label{subsec:coop_vehicle}
Consider a collection of $n$ vehicles moving in a straight line. The state of the $i^{th}$ vehicle is given by its position $x_i(t)$. The $i^{th}$ vehicle has the following measurements:
\setitemize{leftmargin=*, align=left}
\begin{itemize}
   \item {\color{black} Its velocity measurement contaminated by noise, expressed as $ v_{m,i}(t) = v_i(t) + \zeta_i(t),$ where $v_i(t)$ is its velocity and $\zeta_i(t)$ is the noise.}
   
   \item Relative position measurement with others in the collection. Let ${\mathcal S}_i$ denote the set of vehicles with which the $i^{th}$ vehicle can communicate. \textcolor{black}{The relative position information available is represented by $z_{ij}(t) = x_i(t) - x_j(t) + \eta_{ij}(t), \forall \ j \in {\mathcal S}_i,$
   where $z_{ij}(t)$ denotes the relative position measurement between $i^{th}$ vehicle and $j^{th}$ vehicle at time $t$, and $\eta_{ij}(t)$ represents the noise associated with the measurement.} Additionally, for the case of reference vehicle, the relative position measurement is given by $z_{11}(t) = x_1(t) + \textcolor{black}{\eta_{11}(t)}.$
   \item At least one vehicle has an absolute position measurement available; otherwise, it would only be possible to localize relative to each other, but not with respect to a ground frame.
\end{itemize}
We assume that both $\zeta_i(t)$ and $\eta_{ij}(t)$ are independent Gaussian random processes with known statistics. Note that if a measurement $z_{ij}$ is available to the $i^{th}$ vehicle, then the measurement $z_{ji}$ is available to the $j^{th}$ vehicle by assumption. The noise processes across different edges are assumed independent. Given a communication topology, our objective is to find the best possible estimate of the states of the vehicles in the collection, in the least square sense. 

Let $\mathbf{x} \in \mathbbm{R}^n$ represent a vector whose $i^{th}$ component is $x_i$. Then, given the model of the vehicle, $\dot x_i(t) = v_i(t)$, let an observer to estimate the unknowns be: 
\begin{eqnarray*}
\dot {\hat x}_i(t) &=& v_{m,i} (t) +  \sum_{j \in {\mathcal S}_i}K_{ij}z_{ij}(t), i \neq 1, \\
\dot {\hat x}_1(t)&=&v_{m,1}(t) + \sum_{j\in {\mathcal S}_1} K_{1j}z_{1j}(t) + k_0z_{11}(t), 
\end{eqnarray*}
with an associated estimation error given by $\hat x_i(t) - x_i(t)$. 
It is convenient to rewrite the model of the vehicle in terms of the measurement:
$\dot x_i(t) = v_{m,i} - \zeta_i(t), $ and treat $\zeta_i(t)$ as a process noise.
Hence, the state evolution can be compactly represented in the standard form: 
{\color{black}$$\dot {\mathbf{x}}(t) =  \mathbf{0}_{n \times n}{\mathbf{x}(t)} +  \mathbf{I}_n(\mathbf{v}_m(t) +  \boldsymbol{\zeta}(t)). $$}
Let $\mathbf{H}$ represent a concatenation of $n+1$ blocks of row matrices of 
dimension $n$, namely, $\mathbf{H}_i, i=1, 2, \ldots, n+1$.  For $i \le n$, the $j^{th}$  row of $\mathbf{H}_i$ (namely $\mathbf{h}_{ij}$) corresponds to the measurement
$z_{ij}$ as $z_{ij}(t) = \mathbf{h}_{ij}{\mathbf{x}}(t) + \eta_{ij}(t)$. The $(n+1)^{st}$ 
block corresponds to the measurement $z_{11}(t) = \mathbf{H}_{n+1}{\mathbf{x}}(t) + \textcolor{black}{\eta_{11}(t)}$.  In essence, we may express the measurements compactly as ${\mathbf{z}} = \mathbf{H}{\mathbf{x}}(t) + \boldsymbol{\eta}; $ In this case, the dimensions of $\mathbf{z}, \boldsymbol{\eta}$ are equal to $1+\sum_{i=1}^n |{\mathcal S}_i|$.

Now, we can use the {\it Kalman Filter} set up. The Algebraic Riccati equation determines the optimal steady-state filter gain:
$$\mathbf{PF' + FP - PH'R}^{-1}\mathbf{HP +GQG'} = 0, $$ 
where $\mathbf{P}$ is the covariance of the state estimation error, $\mathbf{R}$ is the covariance of the sensor noise $\boldsymbol{\eta}$, $\mathbf{Q}$ is the covariance of the process noise $\boldsymbol{\zeta}$. \textcolor{black}{Since $\mathbf{R}$ represents the covariance matrix of independent Gaussian random processes, it is diagonal; the $j^{th}$ component of the $i^{th}$ block corresponds to the covariance $\frac{1}{C_{ij}}$ of the random process $\eta_{ij}(t)$.  The Kalman filter gain, denoted as $\mathbf{K}$, is given by $\mathbf{K = PH'R}^{-1}$. Since $\mathbf{F} = \mathbf{0}_{n \times n}$ and $\mathbf{G} = \mathbf{I}_n$, the Riccatti equation reduces to $$\mathbf{Q-PH'R}^{-1}\mathbf{H P} = 0. $$}
\textcolor{black}{Assuming $\mathbf{L = H'R}^{-1}\mathbf{H}$, the Riccatti equation can be recast as shown
$$\mathbf{Q - P L P} = 0 \ \Rightarrow \ \mathbf{L} = \mathbf{P}^{-1}\mathbf{Q}\mathbf{P}^{-1} \ \Rightarrow \  \mathbf{L}^{-1} = \mathbf{P} \mathbf{Q}^{-1} \mathbf{P}. $$
Here, $\mathbf{L}$ represents the discrete Dirichlet Laplacian, where the edge weight for the edge $\{i,j\}$ is $C_{ij}$, and vehicle 1 serves as the reference vehicle with absolute position information. $\mathbf{L}$ is non-singular if and only if $C_{11} >0${ and the \it information flow} network is connected. For further details of $\mathbf{L}$, readers may refer to \cite{nagarajan2015maximizing}.}
Since $\underset{t \rightarrow \infty}{\lim} E[\mathbf{e}(t)\mathbf{e}(t)'] = \mathbf{P}$ and
\begin{eqnarray*}
\|\mathbf{L}^{-1}\| = \|\mathbf{P}\mathbf{Q}^{-1}\mathbf{P}\| \le \|\mathbf{P}\| \|\mathbf{Q}^{-1}\| \|\mathbf{P}\|, 
\end{eqnarray*}
we obtain the following inequality
\begin{equation}
    \|\mathbf{P}\| \ge \sqrt{ \frac{\|\mathbf{L}^{-1}\|}{\|\mathbf{Q}^{-1}\|}}\nonumber.
\end{equation}

To minimize the covariance in state estimation error, $\|\mathbf{P}\|$, we must reduce $\|\mathbf{L}^{-1}\|$ or increase $\|\mathbf{Q}^{-1}\|$. However, the process noise $\mathbf{Q}$ is not a design parameter;  $\mathbf{L}$ depends on the communication topology and can be chosen to minimize $\|\mathbf{L}^{-1}\|$ is minimized (or equivalently, maximize the smallest eigenvalue, $\lambda_1(\mathbf{L})$ is maximized) when selecting the communication topology. 



Associated with an $n \times n$ Dirichlet Laplacian matrix, one can always construct an $(n+1)\times (n+1)$ Laplacian matrix. This matrix has Dirichlet Laplacian as its leading principal sub-matrix, with other entries chosen to ensure zero row and column sums. 
This is equivalent to assuming that the frame containing the reference vehicle's absolute position is in motion, with an unknown origin. The upper bound for $\lambda_2$ of the Laplacian matrix also serves as the upper bound for $\lambda_1$ of the Dirichlet Laplacian matrix, as per the Courant-Fischer theorem \cite{parlett1998symmetric}. Thus, the methods in this paper are also applicable to the Dirichlet Laplacian problem, optimizing edge selection to maximize $\lambda_1$ while satisfying resource constraints.

\subsection{An associated problem of interest}
\label{subsec:Associated}
An independent problem arises when considering payload and/or cost budget constraints in the choice of sensors on every vehicle in the earlier problem. A simple way to model this additional requirement is to restrict the number of range sensors that can be mounted on each vehicle, effectively limiting the degree of every node in the network. This problem also arises in Free Space Optical (FSO) networks \cite{son2010building}. Thus, this problem aims to design a robust spanning tree with maximum {\ac}, limiting the degree of each node by a fixed value $d$. 


\section{Mathematical Formulations to Maximize {\Ac}}
\label{Sec:algconn_proposed_formulation}

Let $(V, E, \mathbf{w})$ represent a weighted graph/network. Without any loss of generality, we will simplify the problem by allowing at most one edge to be connected between any pair of nodes in the network without any self-loops. Let $n$ represent the number of nodes in the network, given by $|V|$. Let $w_{ij} > 0$ and $x_{ij} \in \{0, 1\}$ represent the edge weight and the binary choice variable for every edge $\{i,j\} \in E$, respectively. Let $\mathbf{x}$ be the vector of choice variables, $x_{ij}$. If $x_{ij} = 1$, it implies that the edge is chosen in the construction of the network; otherwise, it is not. Given a set $S \subset V$, let $\delta(S)$ denote the edges in the cutset of $S$, i.e., $\delta(S) = \{\{i,j\} \in E | i \in S, \ j \in V\setminus S \}$. 

We may define $$L_{ij} = w_{ij}(\mathbf{e}_i - \mathbf{e}_j)\otimes(\mathbf{e}_i - \mathbf{e}_j),$$ and correspondingly, the weighted Laplacian matrix as
$$ \mathbf{L(x)} = \sum_{i<j, \{i,j\} \in E}  x_{ij} L_{ij}.$$

Note that $\mathbf{L(x)}$ is a symmetric PSD matrix for a given network $\mathbf{x}$. Let $\lambda_1 ( = 0) \leqslant \lambda_2 \leqslant \lambda_3 \leqslant \ldots \leqslant \lambda_n$ be the eigenvalues of $\mathbf{L(x)}$ and $\mathbf{v}_1, \mathbf{v}_2, \ldots, \mathbf{v}_n$ be the respective eigenvectors, where $\lambda_2$ and $\mathbf{v}_2$ are known as the {\ac} and Fiedler vector, respectively.

\label{Subsec:form}
The basic problem ($\mathcal{BP}$) can be expressed as
\begin{equation}
		\begin{array}{lll}
	  		\text{($\mathcal{BP}$)}  &\gamma^* = &\max  \lambda_2(\mathbf{L(x)}), \\
			&\text{s.t.} & \sum_{i<j, \; \{i,j\} \in E} x_{ij} \leqslant q, \\
			& & \mathbf{x} \in \{0, 1\}^{|E|},
		\end{array}
		\label{eq:BP}
\end{equation}
where $q$ is some positive integer which is an upper bound on the number of edges to be chosen. Since this is a non-linear binary program, it is paramount to represent this in a tractable form. In the remainder of this section, we present an equivalent MISDP formulation for $\mathcal{BP}$.
\subsection{Mixed integer semi-definite program}
\label{Subsec:F1}
Let $ \mathbf{e}_0 = \frac{1}{\sqrt{n}}\textbf{1}$, such that $\lVert \mathbf{e}_0 \rVert_2 = 1$. Then, $\mathcal{BP}$ in \eqref{eq:BP} can be equivalently expressed as the following MISDP: 
{\color{black}
\begin{equation*}
		\begin{array}{lll}
		    \text{(${\mathcal F}_0$)} &\gamma^* = &\max  \; \; \gamma, \\
		    &\text{s.t.} &  \sum_{ i<j, \; \{i,j \}\in E} x_{ij} L_{ij} \succeq \gamma (\mathbf{I}_n - \mathbf{e}_0 \otimes \mathbf{e}_0),\\
			& &\sum_{i < j, \; \{i,j\} \in E} x_{ij} \leqslant q, \\
			& &{\color{black}\mathbf{x} \in \{0, 1\}^{|E|}.}
		\end{array}
\end{equation*}}
In the above formulation ${\mathcal F}_0$, the first constraint enforces that $\gamma$ is the {\ac} of the network, the second constraint limits the number of chosen edges by budget, and the third enforces the binary nature of edge selection. We denote the feasible set of this formulation as ${\mathcal S}({\mathcal F}_0)$. Proof of the correctness of ${\mathcal F}_0$ can be found in \cite{nagarajan2012algorithms}. 

In what follows, we will focus on weighted spanning trees as feasible solutions, where the optimal solution (given by $\mathbf{x}^*$) will be a spanning tree with maximum {\ac}. We choose spanning trees as they represent minimally connected networks, although the algorithms developed in this paper can also be generalized to other network types. To this end, in $\mathcal{F}_0$, we will set $q = n-1$, and let ${\mathcal T}$ represent the set of all $n^{n-2}$ spanning trees (represented by last three constraints of $\mathcal{F}_0$). For ease of exposition of the remaining parts of this paper, we introduce a lifted PSD matrix, $ \mathbf{W} = \mathbf{L(x)} - \gamma (\mathbf{I}_n - \mathbf{e}_0 \otimes \mathbf{e}_0 )$, and reformulate $\mathcal{F}_0$ as the following MISDP: 
\begin{subequations}
\begin{align}
    \text{(${\mathcal F}_1$)} \quad \gamma^* = & \ \max  \ \gamma, \\  \text{s.t.} \ &  \mathbf{W}  \succeq 0, \label{eq:W_psd}\\
		& W_{ii} = \sum_{\{i,j\} \in E} w_{ij}x_{ij} - \frac{\gamma(n-1)}{n}, \ \forall  i \in V, \label{eq:W_ii} \\
		& W_{ij} = W_{ji} = -w_{ij}x_{ij} + \frac{\gamma}{n}, \ \forall  \{i,j\} \in E, \label{eq:W_ij} \\
		& \mathbf{x}  \in \mathcal T \label{eq:xinT}.
\end{align}
\label{eq:F_1}
\end{subequations}
\vspace{-0.5cm}
\subsection{Relaxation \& eigenvector cuts}
\label{subsec:eigencuts}
It is well-known that using general purpose MISDP solvers, solving the network design problem ($\mathcal{F}_1$) in  \eqref{eq:F_1} is a herculean task, owing to its computational complexity \cite{mosk2008maximum}. Hence, one of the goals of this paper is not only to solve $\mathcal{F}_1$ efficiently but also to obtain tight upper (dual) bounds as this information can be very useful to quantify the quality of obtained lower-bounding feasible (primal) solutions. To this end, we will now discuss a simple cutting plane-based {\oa} (OA) procedure. 

We call a formulation ${\mathcal F}$ a relaxation for ${\mathcal F}_1$ if the feasible set of the former contains that of the latter, i.e.,  ${\mathcal S}({\mathcal F}_1) \subset {\mathcal S}({\mathcal F})$. In the relaxed formulations we consider in this paper, we replace the semi-definite constraint \eqref{eq:W_psd} of ${\mathcal F}_1$ with a set of simpler related sets of linear cutting planes/cuts/valid inequalities. We refer to ${\mathcal S}({\mathcal F})$ as an OA of ${\mathcal S}({\mathcal F}_1)$, or more loosely ${\mathcal F}$ as an OA of ${\mathcal F}_1$. Thus, to solve the MISDP in $\mathcal{F}_1$, one can exploit the maturity of \textcolor{black}{mixed-integer linear program (MILP)} solvers by iteratively refining MILP relaxations (${\mathcal F}$)  in a cutting plane fashion.

It is well-known that the PSD matrix, $\textbf{W}$, can be viewed as the following semi-infinite description: 
$$\textbf{W} \succeq 0 \ \Longleftrightarrow \ \mathbf{v \cdot Wv} \geqslant 0 \ \forall \mathbf{v}  \in \mathbb{R}^{n} $$
Instead of infinite such vectors, by choosing a finite number of $\mathbf{v}$-s, the above description can be viewed as an OA (relaxation) of the PSD constraint. However, in the following lemma \eqref{lem:finite_vs}, we show that only a finite number of such vector $\mathbf{v}$-s is sufficient to exactly reformulate the MISDP in $\mathcal{F}_1$ in to an MILP ($\mathcal{F}_2$) as follows: 
\begin{subequations}
\begin{align}
    \text{(${\mathcal F}_2$)} \quad \gamma_2^* = & \ \max  \ \gamma, \\  \text{s.t.} \ &  \mathbf{v\cdot Wv} \geqslant 0, \ \forall \mathbf{v} \in {\mathcal V}_F, \label{eq:vWv}\\
		& \text{Constraints} \ \eqref{eq:W_ii}, \eqref{eq:W_ij} \ \text{and} \ \eqref{eq:xinT}.
\end{align}
\label{eq:F_1_fiedler}
\end{subequations}
where ${\mathcal V}_F$ represents the set of all Fiedler vectors corresponding to spanning tree networks, i.e., $\mathbf{v}_2\mathbf{(L(x))} \ \forall \mathbf{x} \in \mathcal{T}$.
\begin{lemma}
    Let $(\gamma^*, \mathbf{x}^*)$ and $(\gamma_2^*, \mathbf{x}_2^*)$ be the optimal solutions of $\mathcal{F}_1$ and $\mathcal{F}_2$, respectively. Then, $\gamma^* = \gamma_2^*$, and the associated feasible solutions, $\mathbf{x}^* = \mathbf{x}_2^*$.
    \label{lem:finite_vs}
\end{lemma}
\begin{proof}
Based on the variational characterization of eigenvalues of a real symmetric matrix, $\mathbf{L(\hat{x})}$, we know the following is true (from Courant-Fischer theorem \cite{parlett1998symmetric}):
{\color{black}
\begin{equation}    
    \lambda_2(\mathbf{L(\hat{x})}) = \min_{\mathbf{v} \in \mathbb{R}^n} \ \{\mathbf{v\cdot L(\hat{x})v}: \ \lVert \mathbf{v} \rVert_2  = 1, \textbf{1}\cdot \mathbf{v} = 0\} 
    \label{eq:lambda2}
\end{equation}
}%
where $\mathbf{\hat{x}}$ is any spanning tree in $\mathcal{T}$. Note that the optimal solution of \eqref{eq:lambda2}, say $\mathbf{\hat{v}}$, corresponds to the Fiedler vector of $\mathbf{L(\hat{x})}$. Using this characterization, one can also exactly reformulate problem $\mathcal{F}_1$ into the following bi-level nonlinear problem, where the outer-level maximizes over all possible spanning trees ($\mathbf{x}$), the $\lambda_2(\mathbf{L(x)})$, or the minimum value of the problem in \eqref{eq:lambda2}: 
\begin{subequations}
    \begin{align*}
        \gamma^* & = \lambda_2(\mathbf{L(x^*)}) \\ & = \max_{\mathbf{x} \in \mathcal{T}} \min_{\mathbf{v} \in \mathbb{R}^n} \ \{\mathbf{v\cdot L(x)v}: \ \lVert \mathbf{v} \rVert_2  = 1, \textbf{1}\cdot \mathbf{v} = 0\}.
    \end{align*}
\end{subequations}
However, for a given $\mathbf{x} \in \mathcal{T}$, since the inner minimization problem's global optimal solution is indeed the Fiedler vector (from \eqref{eq:lambda2}), one can further re-formulate the above bi-level problem into a single-level problem as 
\begin{subequations}
    \begin{align*}
        \gamma^* = \lambda_2(\mathbf{L(x^*)}) & = \max_{\mathbf{x} \in \mathcal{T}} \ \{\gamma : \mathbf{v\cdot L(x)v} \geqslant \gamma, \  \forall \mathbf{v} \in \mathcal{V}_F \}.
    \end{align*}
\end{subequations}
Since it is easy to observe that constraint \eqref{eq:vWv} reduces to $\mathbf{v\cdot L(x)v} \geqslant \gamma$ when $\mathbf{v} \in \mathcal{V}_F$, one can prove that $\gamma^* = \gamma^*_2$ and $\mathbf{x}^* = \mathbf{x}_2^*$. 
\end{proof}

In a complete graph, the size of $|{\mathcal V}_F|$ being an exponential number of spanning tree networks can lead to computational intractability of finding the optimal solution for ${\mathcal F}_2$; for this reason, we only consider a subset ${\mathcal V}_R \subset {\mathcal V}_F$ of Fiedler vectors and relax the semi-definite constraint of ${\mathcal F}_1$. 
However, the resulting MILP is an OA for ${\mathcal F}_1$. An optimal solution, $(\mathbf{x}_s, \gamma_s)$, for the OA may not be feasible for the formulation ${\mathcal F}_1$; in such a case, the matrix $\mathbf{\widehat{W}} := \mathbf{L(x}_s) - \gamma_s (\mathbf{I}_n - \mathbf{e}_0 \otimes \mathbf{e}_0)$ is not positive semi-definite, i.e., \textcolor{black}{at least one eigenvalue of $\mathbf{\widehat{W}}$ is negative.}  Consequently, one can find a Fiedler vector $\mathbf{\bar{v}} \in {\mathcal V}_F$ (corresponding to $\mathbf{L(x_s)}$)
such that the constraint $$\mathbf{\bar{v} \cdot W \bar{v}} \ge 0 $$
is violated at $\mathbf{x} = \mathbf{x}_s$. While true for the optimal solution, this valid inequality or cut is violated by the optimal solution for the relaxed problem and is referred to as an ``eigenvector'' cut. The set ${\mathcal V}_R$ is updated by augmenting $\mathbf{\bar{v}}$ to ${\mathcal V}_R$. By doing so, one can iteratively refine the OA eventually leading to an optimal solution.

\section{Upper Bounds on Optimal {\Ac} of Networks}
\label{Sec:UB}
The primary drawback of the cutting planes discussed in Section \ref{subsec:eigencuts} is their \textit{dense} nature in the variables of the $\mathbf{W}$ matrix. 
This often results in slow convergence or, at worst, stalling at larger upper bounds. One approach to alleviate this issue is by leveraging structured sparsity 
in the network, such as replacing the dense PSD constraint \eqref{eq:W_psd} with the requirement that smaller principal sub-matrices are PSD, 
as seen in \cite{fukuda2001exploiting}. Another recent approach involves adding $k$-sparse cuts for a PSD constraint by enforcing a target sparsity on the added 
cut \cite{dey2022cutting,bhela2021efficient}. \textcolor{black}{However, these methods are not directly applicable to the problem addressed in this paper, as identifying a spanning 
tree from a complete graph lacks inherent sparsity structure.} Authors in \cite{hijazi2016polynomial,gopinath2020proving} propose a nonlinear polynomial 
representation of non-negative principal minors, termed as the ``determinant hierarchy'', while such a method can be cumbersome to derive and implement 
for higher-order principal minors.

In this section, we instead propose a hierarchy of upper bounding MILP formulations for the MISDP problem ${\mathcal F}_1$, with feasible sets containing those of ${\mathcal F}_1$. The key idea is to relax the requirement that $\mathbf{W}$ be a PSD matrix, instead requiring only a subset of its smaller principal sub-matrices to be PSD or, equivalently, their corresponding minors to be non-negative
\cite{blekherman2020sparse}. We enforce this requirement by adding eigenvector-based cuts only on those principal sub-matrices via an OA procedure.

\subsection{Principal minor characterization of PSD matrices}
Below are the fundamental definitions and propositions essential for characterizing a PSD matrix. \cite{prussing1986principal}.
\begin{definition}
\label{def1}
Given a real symmetric matrix $\mathbf{W} \in \mathbbm{R}^{n\times n}$, a minor of $\mathbf{W}$ is the determinant of a sub-matrix obtained by choosing only some rows $J_1 \subseteq \{1,\ldots,n\}$ and some columns $J_2 \subseteq \{1,\ldots,n\}$ of $\mathbf{W}$. A principal minor is the determinant of a principal sub-matrix, $[\mathbf{W}]_J$, obtained by choosing the same rows and columns of $\mathbf{W}$, i.e., $J = J_1 = J_2$.
\end{definition}
\begin{definition}
\textcolor{black}{Let $J_m = \{ J \subseteq \{1,\ldots,n\} : |J| = m\}$ represent the set of all subsets of $\{1,\ldots,n\}$ of size $m$.} Then, for a given $\mathbf{W}$ matrix, \textcolor{black}{$[\mathbf{W}]_{J_m}$} represents the set of all principal sub-matrices of size $m \times m$.
\end{definition}
\begin{proposition}
\label{Prop:1}
$\mathbf{W} \in \mathbbm{R}^{n\times n}$ is a PSD matrix if and only if all its $(2^n-1)$ principal sub-matrices are PSD or the associated principal minors are non-negative, that is, $\forall m=1,\ldots ,n$,
\begin{align*}
   & {\mathbf{W} \succeq 0 \ \Longleftrightarrow \ \mathbf{\widehat{W}} \succeq 0,  \ \forall \mathbf{\widehat{W}} \in [\mathbf{W}]_{J_m},} \\ 
   & \hspace{1.3cm} {\Longleftrightarrow \ det(\mathbf{\widehat{W}}) \geqslant 0, \ \forall \mathbf{\widehat{W}} \in [\mathbf{W}]_{J_m}.}
\end{align*}
\end{proposition}

\subsection{Principal minor-based relaxation formulations}
Utilizing proposition (\ref{Prop:1}), one can construct an MILP relaxation of the MISDP formulation $\mathcal{F}_1$ by considering the PSD-ness of principal sub-matrices, $\color{black}{[\mathbf{W}]_{J_m}}$, of a certain size $m$ ($\leqslant n$).  We employ this characterization through an {\oa} (OA) procedure rather than relying on principal minor characterization to avoid dealing with cumbersome nonlinear polynomial constraints arising from sub-matrix determinants. Enforcing PSD-ness on smaller $m\times m$ principal sub-matrices using OA results in sparser linear constraints, a crucial property leveraged by state-of-the-art MILP solvers to significantly enhance performance.

Algorithm \ref{alg:ub_oa} iteratively refines an OA of the MISDP in ${\mathcal F}_1$, incorporating eigenvectors of principal sub-matrices of size $m$ as cutting planes. The algorithm supports the general case where OA cutting planes can be added on all principal sub-matrices of sizes $m \in \mathcal{M}$. {\color{black} It solves a sequence of MILPs $\mathcal{F}^{u}_{m}$ derived by dropping the PSD constraint \eqref{eq:W_psd} in $\mathcal{F}_1$}. The algorithm yields either: \textbf{(a)} an $\varepsilon_{opt}$-optimal solution to the original MISDP (${\mathcal F}_1$), where upper bound ($UB$) and lower bound ($LB$) are within the prescribed relative optimality tolerance, or \textbf{(b)} a valid $UB$ to ${\mathcal F}_1$. Case \textbf{(a)} holds true when $\max(\mathcal{M})$, given by {\color{black} $\{ i \in \mathcal{M} : i > j \ \forall j \in (\mathcal{M}\setminus \{i\}) \}$,} is equal to $n$, implying the PSD-ness of the $\mathbf{W}$ matrix. However, case \textbf{(b)} applies when $\max(\mathcal{M}) < n$, and the optimal solution to $\mathcal{F}^{u}_{m}$ is a valid upper bound as it's feasible set is an OA of the PSD constraint (see proposition \ref{Prop:1}). To ensure PSD-ness of all smaller principal sub-matrices in case \textbf{(b)}, eigenvector-based cuts are added for all non-PSD sub-matrices in step 9 of Algorithm \ref{alg:ub_oa}. \textcolor{black}{In the initialization (step 2), $UB$ is set to $\gamma^u_{sdp}$, obtained by relaxing the $\mathcal{F}_1$'s binary variables and solving the continuous semi-definite program (SDP) to optimality.}

\begin{algorithm}[ht]
\caption{Outer-approximation algorithm via principal minor characterization}
\label{alg:ub_oa}
{
\begin{algorithmic}[1]

\State \textbf{Input}: Graph $(V,E,\mathbf{w})$, A set of sizes of principal sub-matrices $\mathcal{M}$, PSD tolerance $\varepsilon_{psd} > 0$, Relative optimality tolerance $\varepsilon_{opt} > 0$.

\vspace*{0.05in}
\State \textbf{Initialization}: $LB = 0$, $\textcolor{black}{\gamma^u_{m} = \gamma^u_{sdp}}$, $\mathcal{F}^{u}_{m} =$ MILP by dropping the PSD constraint \eqref{eq:W_psd} in $\mathcal{F}_1$. 

\vspace*{0.05in}
\State Solve $\mathcal{F}^{u}_{m}$. Let the optimal solution be $({\gamma^u}^*, \mathbf{W}^*, \mathbf{x}^*)$. 

\vspace*{0.05in}
\State $LB \longleftarrow \lambda_2(\mathbf{L(x^*)})$

\vspace*{0.05in}
\State $\gamma^u_{m} \longleftarrow {\gamma^u}^*$

\For{$m \in \mathcal{M}$}

\vspace*{0.05in}
\State $\textcolor{black}{[\mathbf{W}]^-_m} \longleftarrow \left\{\mathbf{\widehat{W}} \in \textcolor{black}{[\mathbf{W}]_{J_m}: det(\widehat{\mathbf{W}}^{\ast})} \leqslant -\varepsilon_{psd} \right\}$, \textcolor{black}{where $\widehat{\mathbf{W}}^{\ast}$ is the principal sub-matrix of $\mathbf{W}^*$ with identical indices as $\widehat{\mathbf{W}}$.}

\vspace*{0.05in}
\While{$\textcolor{black}{[\mathbf{W}]^-_m} \neq \emptyset$ \textbf{or} $\frac{\gamma^u_{m}-LB}{\gamma^u_{m} + 10^{-6}} > \varepsilon_{opt}$}

\vspace*{0.05in}
\State Update $\mathcal{F}^{u}_{m}$ with the following linear constraints: 
$$\mathbf{v^- \cdot \widehat{W} v^-} \geqslant 0 \quad \forall  \mathbf{\widehat{W}} \in \textcolor{black}{[\mathbf{W}]^-_m},$$
where $\mathbf{v^-}$ is the violated eigenvector corresponding to the smallest (negative) eigenvalue of $\mathbf{\widehat{W}}^*$.

\vspace*{0.05in}
\State Solve $\mathcal{F}^{u}_{m}$. Update $LB$, $\gamma^u_{m}$ and solutions $\mathbf{W}^*$, $\mathbf{x}^*$.

\vspace*{0.05in}
\State Update $\textcolor{black}{[\mathbf{W}]^-_m}$ based on the updated $\mathbf{W}^*$.
\EndWhile
\EndFor

\vspace*{0.05in}
\State Return $\mathbf{x}^*$, lower bound $LB$ and upper bound $\gamma^u_{m}$.

\end{algorithmic}
}
\end{algorithm}

\begin{remark}
In cases where $|\mathcal{M}| > 1$, as input to Algorithm \ref{alg:ub_oa}, the PSD-ness of smaller sub-matrices is implied by the PSD-ness of sub-matrices of size $\max(\mathcal{M})$ (according to Proposition \ref{Prop:1}). Thus, theoretically, enforcing the former requirement may not be necessary. However, in practice, we observe that sparse eigenvector cuts corresponding to smaller sub-matrices, combined with denser cuts of larger sub-matrices, significantly improve the performance of MILP solvers \cite{gurobi}. An example is illustrated when $\mathcal{M} = {2,n}$, as demonstrated in Section \ref{Sec:CR}.
\end{remark}

\begin{remark}
The cutting plane-based OA method, in Algorithm \ref{alg:ub_oa}, ensures convergence, relying on compactness arguments to establish the boundedness of the feasible region of $\mathcal{F}^{u}_{m}$. Specifically, $trace(\mathbf{W})$ can be worst-case upper bounded by a finite value $\left(\sum_{\mathbf{v} \in V} \sum_{\{i,j\} \in \delta(\mathbf{v})} w_{ij} - \hat{\gamma}(n-1) \right)$, where $\hat{\gamma}$ is any connected network's \ac. Thus, the algorithm's convergence is finite (see \cite{lubin2017mixed} for details). 
\end{remark}

\begin{remark}
    In Algorithm \ref{alg:ub_oa}, we have excluded the discussion on topology cuts (from $\mathcal{F}_0$) for constraints \eqref{eq:xinT}. Similar to eigenvector cuts, we incorporate them within this algorithm in a cutting plane fashion only when solution $\mathbf{x}^*$ forms a disconnected network (see \cite{nagarajan2015synthesizing} for further details). 
\end{remark}

\noindent
\textbf{MISOCP Relaxation:} In the case when $\mathcal{M} = \{2\}$ (i.e., $m=2$), instead of enforcing PSD-ness of $2 \times 2$ principal sub-matrices, one can also simply enforce this requirement via the non-negativity of all the corresponding principal minors, i.e., $det(\mathbf{\widehat{W}}) \geqslant 0 \ \forall \mathbf{\widehat{W}} \in {\color{black}[\mathbf{W}]_{J_2}}$. Therefore, using this characterization, the relaxation formulation, $\mathcal{F}^{u}_{2}$, can also be solved as a mixed-integer second-order conic program (MISOCP) by dropping constraint \eqref{eq:W_psd} from $\mathcal{F}_{1}$ and replacing it with following valid second-order conic (SOC) constraints: 
\begin{equation}
\left \lVert \left(
		\begin{array}{c}
        2W_{ij}\\
        W_{ii} - W_{jj}
		\end{array}
  \right)
  \right \rVert_2 \leqslant W_{ii} + W_{jj}, \quad \forall  \{i,j\} \in E.
\label{eq:conic_con}
\end{equation}

Authors in \cite{bertsimas2020polyhedral} observe that the inclusion of SOC constraints of type \eqref{eq:conic_con}, further in an OA form, perform very well in practice, although in the relaxation of continuous SDPs. Hence, for relaxing the MISDP ($\mathcal{F}_1$), we incorporate these SOC constraints based on an OA procedure within the branch-and-cut framework. For this procedure, cuts employed are based on the following lemma.

\begin{lemma}
Since constraints \eqref{eq:conic_con} reduce to the form $W_{ij}^2 \leqslant W_{ii}W_{jj}$, let $f(W_{ij},W_{ii}) = \frac{(W_{ij})^{2}}{W_{ii}}$. Then, for every $\{i,j\} \in E$, constraint \eqref{eq:conic_con} is satisfied ``if and only if'' the following infinite set of linear cuts holds $\forall  W^o_{ii} \in [\underline{W}_{ii}, \overline{W}_{ii}], W^o_{ij} \in [\underline{W}_{ij}, \overline{W}_{ij}] $ \cite{kelley1960cutting}:
\begin{equation}
    \begin{array}{l}
     f(W^o_{ij},W^o_{ii}) + \frac{df(W^o_{ij},W^o_{ii})}{dW_{ij}}(W_{ij} - W^o_{ij}) + \\ \hspace{8em} \frac{df(W^o_{ij},W^o_{ii})}{dW_{ii}}(W_{ii} - W^o_{ii})  \leqslant W_{jj}.
    \end{array}
    \label{eq:OA_minor}
\end{equation}
\end{lemma}
Although the above cuts provide a semi-infinite representation of the SOC constraint, they can  efficiently integrate into the cutting plane framework of Algorithm \ref{alg:ub_oa} when $\mathcal{M} = \{2\}$. Let $\textcolor{black}{[\mathbf{W}]^-_2}$ be the set of all violated $2\times 2$ sub-matrices (in Step 7). To eliminate these SOC infeasible minors evaluated at the solution $\mathbf{W}^*$, upon simplification of constraints \eqref{eq:OA_minor}, the linear cuts added in step 9 of Algorithm \ref{alg:ub_oa} will be of the form: 
\begin{equation}
{\frac{{W^*_{ij}}}{{(W^*_{ii})^2}}} (2 \ W^*_{ii} \ \widehat{W}_{ij} - W^*_{ij} \ \widehat{W}_{ii}) \leqslant \widehat{W}_{jj}, \ \forall  \mathbf{\widehat{W}} \in \textcolor{black}{[\mathbf{W}]^-_2}
     \label{eq:OA_cut}
\end{equation}

The performance of the aforementioned upper bounding formulations is discussed in Section \ref{Sec:CR}.

\section{Degree-constrained Lower Bounding Formulation}
\label{Sec:L}
In $\mathcal{F}_1$, the exponential number of feasible spanning tree networks ($n^{n-2}$) exacerbates the inherent difficulty of solving the MISDP, especially for larger instances. 
To address this challenge, we introduce a simpler MISDP, exploiting the degree of nodes in optimal networks. This formulation's optimal solution serves as a tight lower bound for $\gamma^*$ in $\mathcal{F}_1$. 

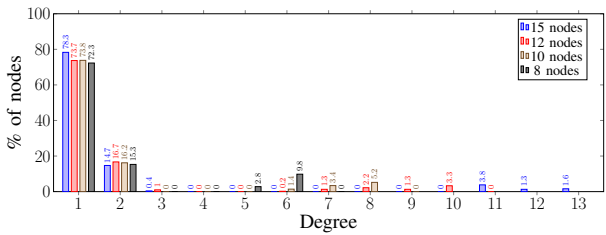
\begin{figure}[ht]
    \centering
    \resizebox{0.9\columnwidth}{!}{%
\begin{tikzpicture}
\begin{axis}[
    ybar=3.25pt,
    bar width=7pt,
    x=1.75cm,
    ymin=0,
    ymax=100,
    xtick=data,
    enlarge x limits= 0.05,
    legend pos=north east,
    legend style={font=\LARGE},
    xlabel={Degree},
    ylabel={\% of nodes},
    nodes near coords,
    every node near coord/.append style={rotate=90, anchor=west},
    label style={font=\Huge},
    tick label style={font=\huge}  
    ]
\addplot coordinates {( 1,78.3	) ( 2,14.7	) ( 3,0.4) ( 4,0) ( 5,0	) ( 6,0	) ( 7,0	) ( 8,0) ( 9,0) ( 10,0) (11, 3.8) (12,1.3) (13,1.6) };
\addplot  coordinates {( 1,73.7) ( 2,16.7) ( 3,1.0) ( 4,0) ( 5,0	) ( 6,0.2) ( 7,1.3	) ( 8,2.2) ( 9,1.3	) ( 10,3.3) (11, 0) };
\addplot coordinates {( 1,73.8	) ( 2,16.2	) ( 3,0	) ( 4,0	) ( 5,0	) ( 6,1.4	) ( 7,3.4	) ( 8,5.2	) ( 9,0	) };
\addplot  coordinates {( 1,72.3) ( 2, 15.3) ( 3,0	) ( 4,0	) ( 5,2.8) ( 6,9.8) (7,0)};

\legend{15 nodes, 12 nodes, 10 nodes, 8 nodes}
\end{axis}
\end{tikzpicture}%
}
\caption{Degree histogram of optimal/best-known solutions for instances with 8, 10, 12 and 15 nodes, averaged over fifty instances of each size. The y-axis denotes the percentage of nodes with a specific degree in an instance.}
\label{Fig:Degree_histogram}
\end{figure}

Fig. \ref{Fig:Degree_histogram} represents the degree histograms for various-sized spanning trees with respect to optimal (8, 10 nodes) and best-known solutions (12, 15 nodes) of $\mathcal{F}_1$. Although this plot does not capture all aspects of robust spanning trees, one can empirically infer insights into their network structures. Optimal trees tend to have very few nodes with significantly higher degrees, resulting in a clustered structure with a low network diameter. Here on, we refer to the node with the highest unweighted degree as the ``central node". For instances with 8, 10, 12 and 15 nodes, the central node's degree lies in the sets \{5,6\}, \{6,7,8\}, \{6,7,8,9,10\},  and \{11,12,13\}, respectively. 

Based on these observations, we now formulate a degree-constrained lower bounding formulation (DCLBF), ${\mathcal F}_{k}^{l}$. In ${\mathcal F}_{k}^{l}$, $1 \leqslant k \leqslant (n-1)$ is a degree-bounding parameter for the central node of all the feasible spanning trees. Since we allow for only one central node with a degree of at least $(n-k)$, we introduce a vector of binary variables, $\mathbf{y}$, whose $i^{th}$ component, $y_i$, is $1$ if the $i^{th}$ node is chosen to be the central node and is $0$ otherwise. Correspondingly, we have the following  formulation:
\begin{subequations}
\begin{align}
    \text{(${\mathcal F}_{k}^{l}$)} \quad {\gamma}^{l^*}_k = & \ \max  \ \gamma, \\  \text{s.t.} \ & \  \text{Constraints} \ \eqref{eq:W_psd}, \eqref{eq:W_ii}, \eqref{eq:W_ij}, \eqref{eq:xinT}, \\
    & {\color{black} \sum_{j \in V\setminus \{i\}}^{n} x_{ij} \geqslant y_{i}(n-k-1)+1,  \forall i \in V,} \label{eq:deg}\\
    &  \sum_{i \in V} y_i = 1, \ \mathbf{y} \in \{0,1\}^{n}.
\end{align}
\label{eq:F_{d_k}}
\end{subequations}
The MISDP formulation ${\mathcal F}_{k}^{l}$ in \eqref{eq:F_{d_k}} can be solved to optimality using the eigenvector-based cutting plane framework in Algorithm \ref{alg:ub_oa}. The only difference would be in step 2, where $\mathcal{F}_1$ will be replaced by ${\mathcal F}_{k}^{l}$ for a certain input value of $k$. 
\begin{remark}
\label{remark:dclbf_bounds}
    In formulation \eqref{eq:F_{d_k}}, increasing the parameter $k$ reduces the lower bound on the central node degree but increases the feasible space, thus raising problem complexity. Further, the optimal objective, ${\gamma}^{l^*}_k$, monotonically increases with the degree of the central node, i.e., ${\gamma}^{l^*}_1  \leqslant {\gamma}^{l^*}_2 \leqslant \ldots \leqslant {\gamma}^{l^*}_{n-1} = \gamma^*$, where $\gamma^*$ is the optimal objective of $\mathcal{F}_1$ in \eqref{eq:F_1}.
\end{remark}

\begin{remark}
    In cooperative vehicle localization with noisy measurements (from Section \ref{subsec:coop_vehicle}), to minimize the state estimation error, it is reasonable to assume that most vehicles will be connected to the vehicle with the absolute position measurement available, acting as a central node. Thus, the underlying communication network can be assumed to resemble the optimal solution of ${\mathcal F}_{k}^{l}$ at lower values of $k$.
\end{remark}

\section{Maximum Cost Heuristic}
\label{Sec:H}
This section presents a maximum cost heuristic (MCH) aimed at efficiently obtaining high-quality feasible solutions. This heuristic exploits the empirical trend in topological structures of optimal networks of ${\mathcal F}_k^l$ in \eqref{eq:F_{d_k}}:
\begin{itemize}[leftmargin=*]
    \item {\bf Observation about central node:} Given a fixed value of $k$, in most instances, the sum of weights of the $(n-k)$ edges incident on the central node exceeds the corresponding value of any other node. Based on this metric, this observation allows us to  form the priority order, $\mathbf{O}^*_{cn}$. Analysis of Table \ref{Tab:CN_data}, suggests that the central node of the optimal network of ${\mathcal F}_k^l$ often ranks within the top $h_1$ elements of $\mathbf{O}^*_{cn}$, where $h_1$ denotes the number of available choices for a central node from the priority order.
    \item {\bf Observation about the edges connecting leaf nodes:} Often, the edges found in the optimal network of   ${\mathcal F}_k^l$ correspond to those with a higher value of $w_{ij}\cdot(v_i-v_j)^2$, where $v_i$ is the $i^{th}$ component of the Fiedler vector of the star graph with the central node chosen based on the first observation, and the edge $(i,j)$ connects the leaf node $j$ with a node $i$ that is connected to the central node. We assign a ranking to the edges of each leaf node accordingly in the priority order, $\mathbf{O}^*_{le}$, where higher edge scores indicate better ranking. Analysis from Table \ref{Tab:CN_data} suggests that these edges often occupy the top $h_2$ elements of $\mathbf{O}^*_{le}$, where $h_2$ denotes the number of available choices for connecting edges from $\mathbf{O}^*_{le}$ to each leaf node.
\end{itemize}

\begin{table}[ht]
\centering
    \caption{Percentage of 50 instances where the central node and the edges connecting leaf nodes of the optimal network of ${\mathcal F}_k^l$ lie in the first $h_1$ elements of $\mathbf{O}^*_{cn}$ and the first $h_2$ elements of $\mathbf{O}^*_{le}$, respectively, for varying $n$, $k$, $h_1$ and $h_2$ values.}
    \label{Tab:CN_data}
    \resizebox{0.9\columnwidth}{!}{%
    \begin{tabular}{p{0.5cm} p{0.5cm} | p{0.75cm} p{0.75cm} p{0.75cm}| p{0.75cm} p{0.75cm} p{0.75cm}} 
    \Xhline{2\arrayrulewidth} \\ [-1.2em]
   $n$ & $k$ &  & $h_1$ &  &  & $h_2 $ & \\ [0.5ex]
    \hline \\ [-1.2em]
     &  & \ 3 & \ 5 & \ 7 & \ 3 & \ 5 & \ 7\\
    \hline \\ [-1.2em]
    8 & 3 & 0.94 & 0.98 & 0.98 & 0.96 & 1.00 & -- \\
    10 & 4 & 0.80 & 0.90 & 0.98 & 0.96 & 1.00 & --\\
    12 & 5 & 0.76 & 0.88 & 0.96 & 0.97  & 0.99 & 1.00\\
    15 & 4 & 0.60 & 0.72 & 0.84 & 0.91 & 0.99 & 1.00\\ [0.5ex]
    \Xhline{2\arrayrulewidth}
    \end{tabular}%
    }
\end{table}

Algorithm \ref{alg:H} outlines the process to generate the ranking orders of nodes and edges based on the observations mentioned. The first part of the Algorithm \ref{alg:H} (\textit{lines 3:8}) gives the ranking of nodes ($\mathbf{O}^*_{cn}$) for them to be considered a central node;  ranking is based on the sum of the weights of the heaviest $(n-k)$ edges incident on each node $i\in V$ in a complete graph. The latter part of the Algorithm \ref{alg:H} (\textit{lines 9:22}) provides the ranking of edges connecting the leaf nodes ($\mathbf{O}^*_{le}$).

\begin{algorithm}[ht]
\begin{algorithmic}[1]
\State \textbf{input}: $n,k,\mathbf{w},h_1$.

\vspace*{0.05in}
\State \textbf{initialization}: $\mathbf{O}^*_{cn} \gets \phi_{1 \times n}$, $\mathbf{O}^*_{le} \gets \phi_{k-1 \times n-1}$.

\vspace*{0.05in}
\For{$i \in V$}

\vspace*{0.05in}
\State $\mathbf{C}[i,:] \gets \text{sort}(\mathbf{w}[i,:])$ \Comment{Descending order}

\vspace*{0.05in}
\State Store indices of $\mathbf{C}[i,:]$ in $ \mathbf{C^*}$, such that $\mathbf{C}[i,:]$ is sorted in decreasing order.

\vspace*{0.05in}
\State $\mathbf{S}[i] \gets \text{sum}(\mathbf{C}[i,j] \ for \ \color{black}{j = 1,2,\ldots,n-k})$

\EndFor

\State Store indices of $\mathbf{S}$ in $ \mathbf{O}^*_{cn}$, such that $\mathbf{S}$ is sorted in decreasing order.

\vspace*{0.05in}
\For{\textcolor{black}{$i = 1,2,\ldots,h_1$}}

\vspace*{0.05in}   
\For{\textcolor{black}{$j = (n-k+1),\ldots,n$}}

\vspace*{0.05in}
\State Construct a star graph with $\mathbf{O}^*_{cn}[i]$ as central node.

\vspace*{0.05in}
\State Compute Fiedler vector $\mathbf{v}$ of star graph.

\vspace*{0.05in}
\For{\textcolor{black}{$l = 1,2,\ldots,n $}}

\vspace*{0.05in}
\If{\textcolor{black}{$l$ is not $\mathbf{O}^*_{cn}[i]$}}

\vspace*{0.05in}
\State $\mathbf{O_{le}}[i,j,l] \gets \mathbf{w}[j,l]*(\mathbf{v}[\mathbf{C^*}[j]] - \mathbf{v}[l])^2$

\vspace*{0.05in}
\State Store indices of $\mathbf{O_{le}}[i,j,:]$ in $\mathbf{O}^*_{le}[i,j]$, such that $\mathbf{O_{le}}[i,j,:]$ is sorted in decreasing order.

\EndIf
\EndFor
\EndFor
\EndFor

\vspace*{0.05in}
\State Return $\mathbf{O}^*_{cn}$ and $\mathbf{O}^*_{le}$.
\caption{Ranking algorithm for the weighted network}
\label{alg:H}
\end{algorithmic}
\end{algorithm}

For chosen values of $k$, $h_1$ and $h_2$, using the priority orders generated via Algorithm \ref{alg:H}, additional constraints are added to ${\mathcal F}_{k}^{l}$ (we will refer to it as ${\mathcal F}_{k}^{h}$) to limit the feasible choices for the central node and the edges connecting the leaf nodes. Solving the MISDP in ${\mathcal F}_{k}^{h}$ via cutting plane algorithm, as described in Algorithm \ref{alg:ub_oa}, by setting $\mathcal{M} = \{n\}$ results in a good quality feasible solution, whose {\ac} will be referred as $\gamma_h$. The quality of the $\gamma_h$ can be improved by increasing the $k$, $h_1$, and $h_2$, while on the other hand, the runtime also increases. The quality and scalability of the MCH are discussed in Section \ref{Sec:CR}.

\subsection{Maximum cost heuristic for the associated problem}
For the problem described in Section \ref{subsec:Associated}, the degree of the central node is upper bounded by $d$ in contrast to the DCLBF ${\mathcal F}_{k}^{l}$, where the degree is lower bounded by $(n-k)$ \eqref{eq:deg}. The MISDP formulation for these networks differs from ${\mathcal F}_{k}^{l}$ with respect to the following degree constraint:

\begin{equation}
     \sum\nolimits_{j=1}^{n} x_{ij} \leqslant y_{i}d,\ \forall \ i = 1,2,\ldots,n.
     \label{eq:deg_con_payload}
\end{equation}

Based on the MISDP formulation for the networks with payload constraints, the MCH has been modified accordingly. We generate the priority orders $\mathbf{O}^*_{cn}$ and $\mathbf{O}^*_{le}$ using $d$ instead of $(n-k)$. In the forthcoming section, we corroborated the performance of the modified MCH algorithm with the FSM algorithm from \cite{son2010building}.

\section{Computational Results}
\label{Sec:CR}
All optimization formulations and algorithms were implemented using JuMP v1.2.0 \cite{dunning2017jump} in Julia v1.7.3 programming language. The code is accessible via the open-source Julia package  ``\textsc{LaplacianOpt}"\footnote{\label{note:lopt} \url{https://github.com/harshangrjn/LaplacianOpt.jl}}. All computational results were computed with Gurobi 9.5.1 \cite{gurobi} as an MILP solver and Mosek 9.2.16 \cite{mosek_2020} as the convex SDP solver on a personal laptop with 2.9 GHz 6-Core Intel Core i9 processor and 16GB memory. User-defined cuts, such as the eigenvector and topology cuts (in Algorithm \ref{alg:ub_oa}), were implemented using Gurobi lazy-cut callback to separate integral solutions. \textcolor{black}{In this section, boldface is used in tables to emphasize either the best outcomes achieved by the proposed algorithm or the improved performance compared to existing methods.}

\subsection{Instance generation}
\textcolor{black}{The proposed algorithms were evaluated across instances ranging from eight to one hundred nodes. 
To ensure the non-triviality of the optimal networks, each instance was deliberately chosen where there is at least one feasible solution with {\ac} greater than that of star graphs and the maximum spanning tree.} All test instances are included in the ``\textsc{LaplacianOpt}"\footref{note:lopt} package.

\subsection{Performance of Algorithm \ref{alg:ub_oa}}
\label{subsec:UB_perf}
\subsubsection{\underline{Quality of upper bounds}} As discussed in Section \ref{Sec:UB}, \textcolor{black}{executing Algorithm \ref{alg:ub_oa}} with $\max(\mathcal{M}) < n$ guarantees an upper bound to $\mathcal{F}_1$. The performance of the proposed upper bounding algorithm is demonstrated for various instances of each problem size, utilizing the principal minors of sub-matrices with sizes $2\times 2, 3\times 3,$ and $4\times 4$ in Table \ref{Tab:UB}.

In Table \ref{Tab:UB}, the optimality gap is defined by $ (\frac{\gamma^u - \gamma^{l}}{\gamma^l} * 100),$ where $\gamma^u$ is the upper bound attained and $\gamma^l$ is equal to $\lambda_2(\mathbf{L(x^*}))$ of the optimal/best-known feasible solution $\mathbf{x^*}$.  The best-found $\gamma^l$ values are provided in Table \ref{tab:optimal_ac} in the Appendix. \textcolor{black}{$\gamma^u_{sdp}$ is obtained by relaxing binary variables within $\mathcal{F}_1$ and solving the resulting continuous SDP problem to optimality. Conversely, $\gamma^u_{m}$ is computed using Algorithm \ref{alg:ub_oa} for principal sub-matrices of size $m$.} For instances with eight and ten nodes, the best-known feasible solution corresponds to the optimal solution obtained by solving ${\mathcal F}_{1}$ using the cutting plane algorithm, with $\mathcal{M} = \{n\}$. For twelve node instances, the best-known feasible solution is obtained by solving ${\mathcal F}_{k}^{l}$ in \eqref{eq:F_{d_k}} with $k$ set to five.

\begin{table}[ht]
    \centering
    \caption{Comparing the optimality gaps between upper bounds and optimal/best-known feasible solutions for varying sizes of networks and principal sub-matrices.}
	\subtable[$n = 8$]{
    \centering
    \label{Tab:up_bnd8}
    \resizebox{0.75\columnwidth}{!}{%
    \begin{tabular}{l * {4}{c}} 
    \Xhline{2\arrayrulewidth} \\ [-0.75em]
   Instance & $\gamma_{sdp}^u$ gap (\%) & $\gamma^u_{2}$ gap (\%) & $\gamma^u_{3}$   gap (\%)  & $\gamma^u_{4}$  gap (\%)   \\ [0.5ex]
    \hline \\ [-0.5em]
    1 & 105.91 &59.11 &  {15.63} &  \textbf{0.01} \\
    2 & 132.15 &38.53 &  {18.06} &  {0.02}  \\
    3 & 130.00 &68.79 &  {39.52} &  {0.37}  \\
    4 & 127.93 &54.03 &  {16.90} &  {0.21}  \\
    5 & 118.82 &64.59 &  \textbf{0.50}  &  {0.14} \\
    6 & 130.66 &55.76 &  {8.06}  &  {0.87}  \\ 
    7 & 136.94 &58.35 &  {22.38} &  {0.36}  \\
    8 & 113.15 &49.45 &  {7.84}  &  {0.30}  \\
    9 & 126.67 &43.22 &  {20.60} &  {0.14}  \\
    10 & 106.41 &\textbf{38.33} &  {22.55} &  {3.90}  \\ [1ex]
    \Xhline{2\arrayrulewidth} \\ [-0.5em]
    Average & 122.82 &\textbf{53.01} &  \textbf{17.20} &  \textbf{0.63}  \\ [1ex]
    \Xhline{2\arrayrulewidth}
    \end{tabular}%
    }
    }
    \subtable[$n = 10$]{
    \centering
    \label{Tab:up_bnd10}
    \resizebox{0.75\columnwidth}{!}{%
    \begin{tabular}{l *{4}{c}} 
    \Xhline{2\arrayrulewidth} \\ [-0.75em]
   Instance & $\gamma_{sdp}^u$ gap (\%) & $\gamma^u_{2}$ gap (\%) & $\gamma^u_{3}$   gap (\%)  & $\gamma^u_{4}$  gap (\%)   \\ [0.5ex]
    \hline \\ [-0.5em]
    1 & 216.34 & 103.01  & {48.97} &  {18.87} \\
    2 & 170.24 &   83.87  & {36.60} &  {6.76} \\
    3 & 188.82 &   70.67  & {39.25} &  {6.45} \\
    4 & 146.74 &   54.41  & {15.31} &  \textbf{0.20} \\
    5 & 193.24 &   109.56  & {43.83} & {13.63} \\
    6 & 112.88 &   \textbf{46.03}  & \textbf{12.34} &  {3.49} \\ 
    7 & 213.73 &   85.69  & {45.59} &  {19.22} \\
    8 & 168.55 &   66.84  & {29.73} &  {0.90} \\
    9 & 170.00 &   73.23  & {23.96} &  {7.06} \\
   10 & 204.16 &  70.51  & {35.72} & {28.10} \\ [1ex]
   \Xhline{2\arrayrulewidth} \\ [-0.5em]
    Average & 178.47 & \textbf{76.38} &  \textbf{33.13} &  \textbf{10.47}\\ [1ex]
    \Xhline{2\arrayrulewidth}
    \end{tabular}%
    }
    }
    \subtable[$n = 12$]{
    \centering
    \label{Tab:up_bnd12}
    \resizebox{0.62\columnwidth}{!}{%
    \begin{tabular}{l*{3}{c}} 
    \Xhline{2\arrayrulewidth} \\ [-0.75em]
   Instance & $\gamma_{sdp}^u$ gap (\%) & $\gamma^u_{2}$ gap (\%) & $\gamma^u_{3}$   gap (\%)  \\ [0.5ex]
    \hline \\ [-0.5em]
    1 & 167.13 &111.33 &  47.80  \\
    2 & 185.42 &87.61 &  68.12 \\
    3 & 202.65 &102.34 &  64.42 \\
    4 & 204.27 &122.34 &  64.83  \\
    5 & 178.62 &74.29  &  37.75 \\
    6 & 182.15 &94.03 &  48.05 \\
    7 & 141.15 &\textbf{59.91} &  \textbf{29.54} \\
    8 & 223.04 &118.75 &  64.48 \\
    9 & 173.84&82.05 & 43.89 \\
    10 & 154.89 &98.19 & 51.36 \\ [1ex] 
    \Xhline{2\arrayrulewidth} \\ [-0.5em]
    Average & 181.32 &\textbf{95.04} &\textbf{52.02}   \\ [1ex]
    \Xhline{2\arrayrulewidth}
    \end{tabular}%
    }
    }
    \label{Tab:UB}
\end{table}

In  Tables \ref{Tab:up_bnd8}, \ref{Tab:up_bnd10}, and \ref{Tab:up_bnd12}, the best gaps obtained for instances of eight, ten, and twelve nodes via $\mathcal{F}^u_{2}$ relaxation are \textbf{38.33\%} (106.41\%), \textbf{46.03\%} (112.88\%), and \textbf{59.91\%} (141.15\%), respectively, where the values within parenthesis represent binary relaxation gaps. Similarly, the best gaps obtained for instances of eight, ten, and twelve nodes via ${\mathcal F}^u_{3}$ relaxation are \textbf{0.50\%} (118.82\%), \textbf{12.34\%} (112.88\%), and \textbf{29.54\%} (141.15\%). In case of ${\mathcal F}^u_{4}$ relaxation, we obtain \textbf{0.01\%} (105.91\%) and \textbf{0.20\%} (146.74\%) as best gaps for instances of eight and ten nodes, respectively. However, ${\mathcal F}^u_{4}$ relaxation for instances of twelve nodes times out (1 hr. wall time limit). \textcolor{black}{The upper bounds achieved by the proposed algorithm are significantly better than easy-to-find $\gamma^u_{sdp}$ values, as observed in the `Average' row in Tables \ref{Tab:up_bnd8}, \ref{Tab:up_bnd10}, and \ref{Tab:up_bnd12}.} As anticipated, including larger principal sub-matrix cuts via OA reduces upper bound gaps, albeit with longer runtimes.

\subsubsection{\underline{Runtimes for solving $\mathcal{F}_1$ to optimality}} 
The runtime for obtaining optimal solutions of $\mathcal{F}_1$ is notably reduced for medium-sized instances by coupling sparse eigenvector cuts for smaller sub-matrices with denser cuts of larger sub-matrices. \textcolor{black}{In Table \ref{Tab:CT_comp}, we compare the runtimes across nine and ten node instances using eigenvectors of various principal sub-matrix size sets as cutting planes, including $\mathcal{M} = \{2,n\}, \{3,n\}, \{4,n\}, \{2,3,n\}, \{2,4,n\},$ and $\{3,4,n\}$ within Algorithm \ref{alg:ub_oa}. These results are contrasted with the cutting plane algorithm based on eigenvectors of size $n$ (i.e., $\mathcal{M} = \{n\}$ within Algorithm \ref{alg:ub_oa}), shown in column `$\{n\}$' in Table \ref{Tab:CT_comp}.}

\begin{table}[ht!]
\centering
    \caption{\textcolor{black}{Comparing runtimes of Algorithm \ref{alg:ub_oa} for $\mathcal{F}_1$ using eigenvector cuts of various principal sub-matrix sizes with an existing algorithm based on eigenvector cuts of only size $n$. This comparison is across instances with nine and ten nodes.}}
    \label{Tab:CT_comp}
    \resizebox{\columnwidth}{!}{%
    \begin{tabular}{p{0.5cm} p{1.75cm} p{1cm} p{1cm} p{1cm} p{1cm} p{1cm} p{1cm} p{1cm}} 
    \Xhline{2\arrayrulewidth} \\ [-1em]
   $n$ & Run time (s) &  &  &  & $\mathcal{M}$ &  &  & \\ [0.5ex]
    \hline \\ [-1em]
       &  &  $\{n\}$ & $\{2,n\}$ & $\{3,n\}$ & $\{4,n\}$ &  $\{2,3,n\}$ &  $\{2,4,n\}$ &  $\{3,4,n\}$\\ [0.5ex]
    \hline \\ [-0.5em]
    9 & Average & 48.6 & \textbf{35} & \textbf{38.1} & 122.4 & \textbf{48.1} & 113.7 & 147.8  \\
      & Minimum & 13.6 & \textbf{9.7} & \textbf{11.8} & 41.4 & 14.9 & 31.8 & 28.5  \\
      & Median & 34.0 & 34.8 & \textbf{31.7} & 93.1 & 36.5 &  94.8 & 105.3\\
      & Maximum & 223 & \textbf{102.9} & \textbf{106.9} & 329.1 & \textbf{184.7} & 310.8 & 510.0  \\[1ex]
    \hline \\ [-0.5em]
   10 & Average & 1179.3 & \textbf{838.7} & \textbf{612.2} & 1797.6 & \textbf{770.0} & 1906.8 & 2188.2  \\
      & Minimum & 35.1 & \textbf{35.1}    & 42.9  & 116.3 & 44.6 & 121.3 &  176.8 \\
      & Median & 515.6 & \textbf{354.9}   & \textbf{386.0} & 1091.1 & \textbf{426.9} & 1167.2 &  1479.3 \\
      & Maximum & 5337.4 & 6359.3 & \textbf{2386.3} & 7467.2 & \textbf{3783.6} & 6911.7 &  9853.6 \\
  [1ex]
    \Xhline{2\arrayrulewidth}
    \end{tabular}%
    }
\end{table}

Table \ref{Tab:CT_comp} reveals a significant decrease in average runtimes for obtaining optimal solutions when including eigenvectors of principal sub-matrices of sizes $2\times 2$ and $3\times 3$ (i.e., $\mathcal{M} = \{2,n\}, \{3,n\}$, and $\{2,3,n\}$). However, for cases where $\mathcal{M} = \{4,n\}, \{2, 4,n\}$, and $\{3,4,n\}$, the runtimes increase due to the verification of higher number of $4 \times 4$ principal sub-matrices. Moreover, the corresponding cuts are denser than those for the smaller $2\times 2$ and $3\times 3$ principal sub-matrices.

\subsection{Performance of the DCLBF (${\mathcal F}_{k}^{l}$)}

\subsubsection{\underline{DCLBF solutions}} For problem instances larger than ten nodes, the optimal solutions of $\mathcal{F}_1$ are unknown. Utilizing the DCLBF ${\mathcal F}_{k}^{l}$ in \eqref{eq:F_{d_k}}, we can obtain reasonable lower bounds in significantly less runtime for larger instances.

\begin{figure}[ht]
    \centering
	\subfigure[{$n = 15, k = 5  $}]{
	\includegraphics[scale=0.34]{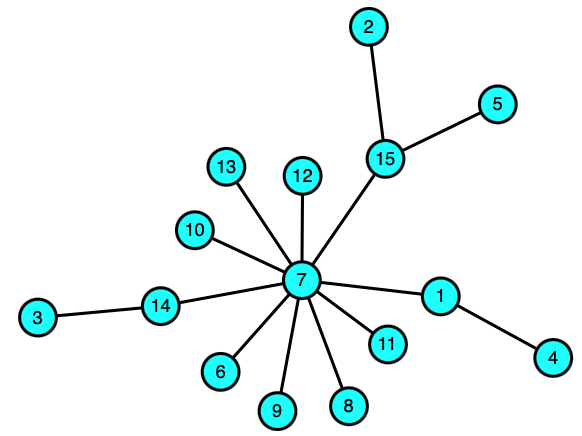}}
	\subfigure[{$n = 25, k = 9 $}]{
	\includegraphics[scale=0.255]{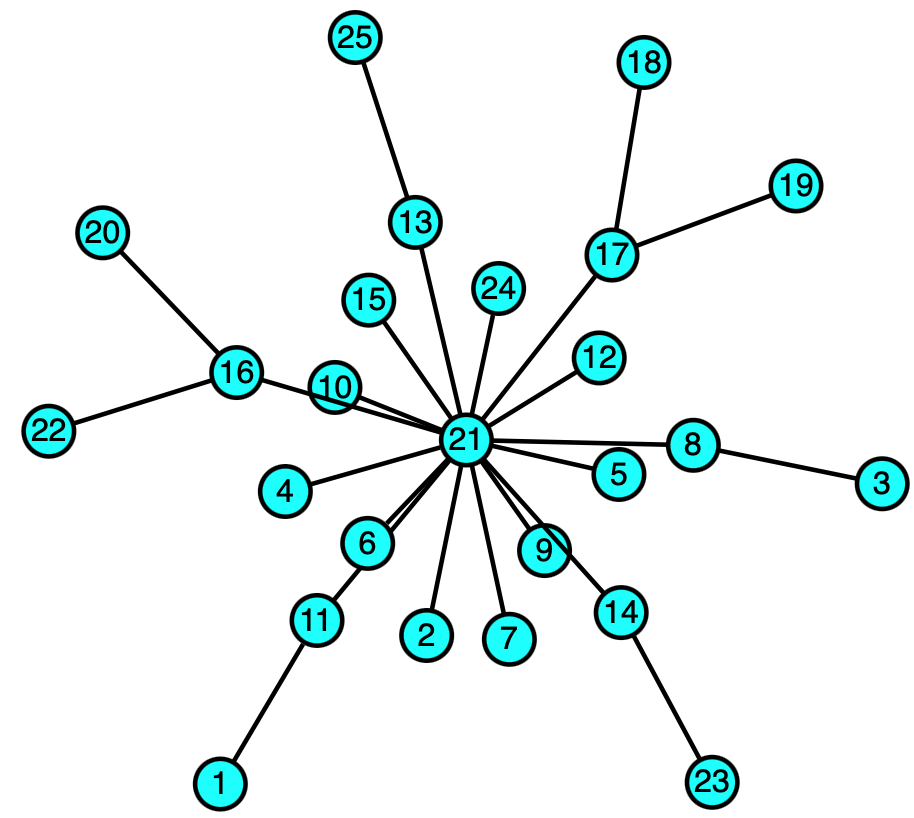}}
    \subfigure[{$n = 60, k = 27 $}]{
	\includegraphics[scale=0.235]{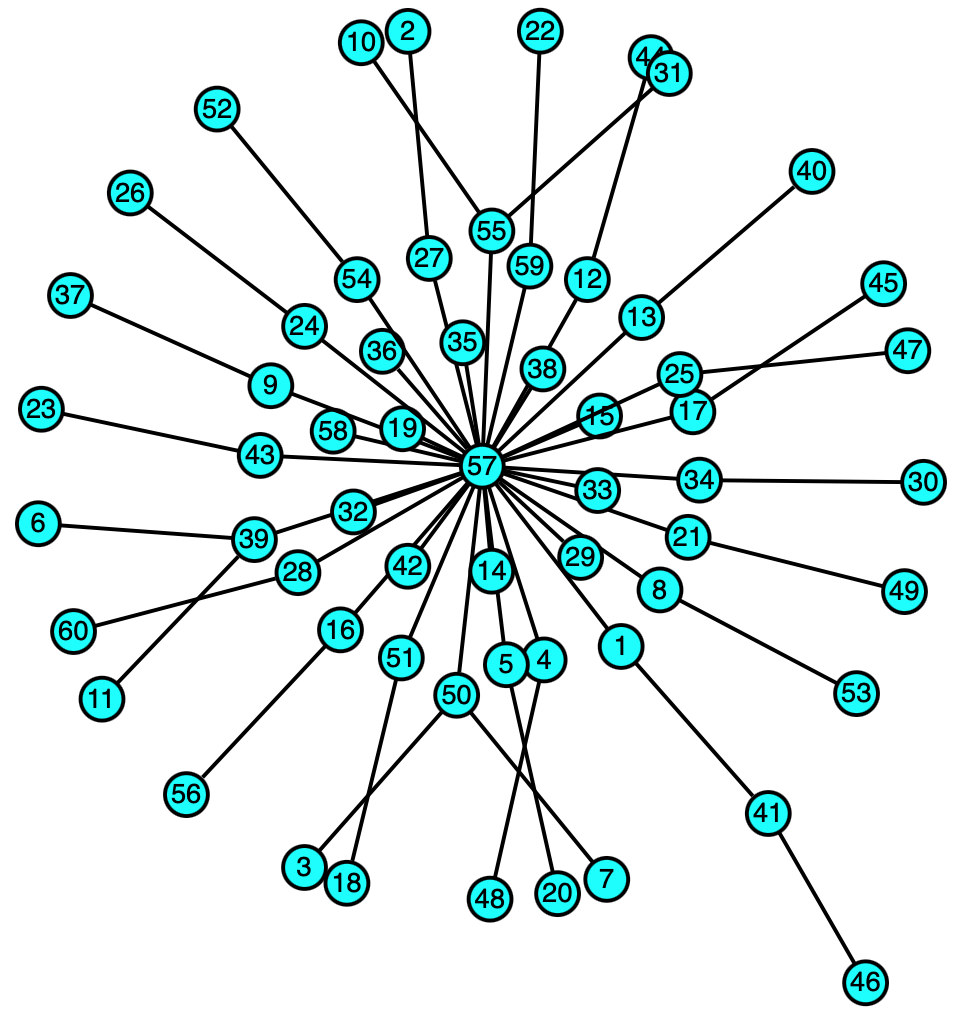}}
    \subfigure[{$n = 100, k = 16 $}]{
	\includegraphics[scale=0.25]{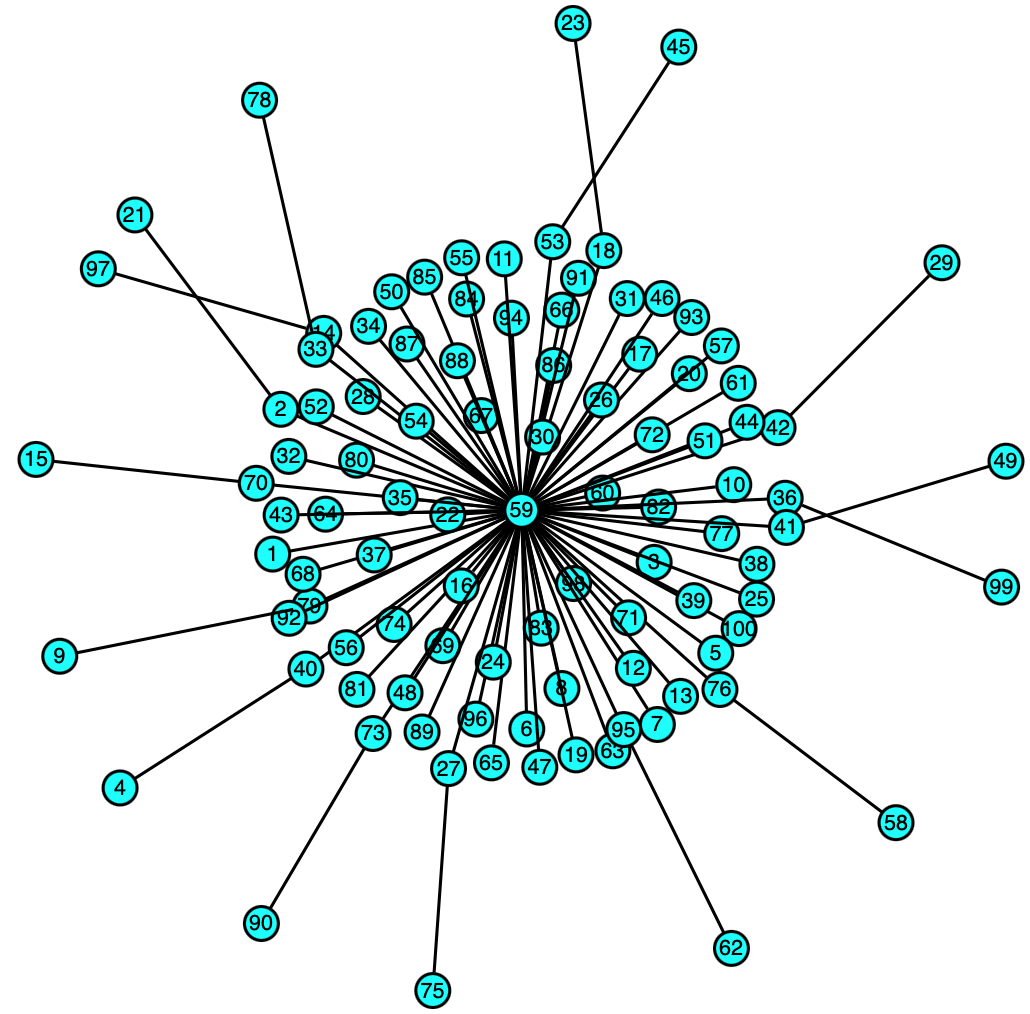}}
    \caption{\textcolor{black}{The best known feasible networks for  varying instance sizes obtained using DCLBF ${\mathcal F}_{k}^{l}$ in \eqref{eq:F_{d_k}} with various $k$ values. Edge weights are omitted for clarity.}}
    \label{Fig:Feasible_networks}
\end{figure}

For a subset of instances larger than ten nodes, the best-known feasible solutions are presented in Fig. \ref{Fig:Feasible_networks}, computed in an average time of less than {\it thirty minutes}.

\subsubsection{\underline{Performance of DCLBF for varying $k$ values}}
In the DCLBF \eqref{eq:F_{d_k}}, the quality of the lower bounding solution and its runtime depends on the value of $k$ chosen, as the size of the feasible set changes with $k$. Fig. \ref{Fig:k} compares the solution quality and runtimes for different $k$ values for ten-node instances. The solution quality improves for all instances of ten nodes with increasing $k$. \textcolor{black}{As discussed in Remark \ref{remark:dclbf_bounds}, optimal solutions (as shown in Table \ref{tab:optimal_ac}) are achieved when the DCLBF is computed at $k=4$ in \eqref{eq:F_{d_k}}, leading to significantly reduced runtimes compared to those of $\mathcal{F}_1$ in Table \ref{Tab:CT_comp}. A similar trend was observed for larger instances.}

\begin{figure}[ht]
    \centering
    \resizebox{0.8\columnwidth}{!}{%
\begin{tikzpicture}
\begin{axis}[
    x tick label style={/pgf/number format/1000 sep=},
    ybar=2.5pt,
    bar width=6pt,
    x=1.5cm,
    ymin=0,
    ymax=45,
    xtick=data,
    enlarge x limits=0.07,
    ylabel={Algebraic connectivity $\gamma^{l}_{k}$},
    ymajorgrids,
    label style={font=\huge},
    tick label style={font=\Large}  
    ]
\addplot  coordinates {( 1,23.8864) ( 2,	18.0900	) ( 3,28.6787	) ( 4,14.4100	) ( 5,12.0749	) ( 6,23.8447	) ( 7,32.5620	) ( 8,34.1772	) ( 9,12.8389	) ( 10,16.6890)};
\addplot coordinates {( 1,31.1254	) ( 2,38.9524	) ( 3,37.1088	) ( 4,31.4042	) ( 5,26.1102	) ( 6,24.3869	) ( 7,32.6700	) ( 8,34.8590	) ( 9,39.4034	) ( 10,31.9314)};
\addplot  coordinates {( 1,34.2371	) ( 2,41.4488	) ( 3,37.7309	) ( 4,41.4618	) ( 5,34.3193	) ( 6,39.9727	) ( 7,36.1651	) ( 8,34.8590	) ( 9,39.4034	) ( 10,33.6339)};
\addplot coordinates {( 1,34.2371	) ( 2,41.4488	) ( 3,37.7309	) ( 4,41.4618	) ( 5,34.3193	) ( 6,39.9727	) ( 7,36.1651	) ( 8,42.3291	) ( 9,39.4034	) ( 10,34.9161)};
\end{axis}
\end{tikzpicture}%
}
\resizebox{0.8\columnwidth}{!}{%
\begin{tikzpicture}
\begin{axis}[
    x tick label style={/pgf/number format/1000 sep=},
    ybar=2.5pt,
    bar width=6pt,
    x=1.5cm,
    ymin=0,
    ymax=19,
    xtick=data,
    enlarge x limits=0.07,
    legend style={at={(0.5,-0.25)}, anchor=north,legend columns=-1, style={column sep=0.4cm}, font=\Large},
    xlabel={Instance no.},
    ylabel={Run time (s)},
    ymajorgrids,
    label style={font=\huge},
    tick label style={font=\Large}  
    ]
\addplot coordinates {( 1,0.7921) ( 2,0.8789) ( 3,0.7419) ( 4,0.7323) ( 5,0.8128) ( 6,0.7823) ( 7,0.7569) ( 8,0.7759) ( 9,0.7340) ( 10,0.8919)};
\addplot coordinates {( 1,0.8497) ( 2,0.8812) ( 3,0.8461) ( 4,0.9009) ( 5,0.8262) ( 6,0.8341) ( 7,0.9637) ( 8,0.9988) ( 9,0.9068) ( 10,0.9158)};
\addplot  coordinates {( 1,3.1154) ( 2,1.3358) ( 3,3.0079) ( 4,1.3069) ( 5,2.0901) ( 6,1.0573) ( 7,1.4837) ( 8,3.2635) ( 9,2.6844) ( 10,3.3833)};
\addplot coordinates {( 1,15.1147) ( 2,18.0393) ( 3,17.0430) ( 4,5.1511) ( 5,4.4396) ( 6,4.0093) ( 7,16.1823) ( 8,12.7117) ( 9,12.5035) ( 10,13.7834)};

\legend{$k$ = 1, $k$ = 2, $k$ = 3, $k$ = 4}
\end{axis}
\end{tikzpicture}%
}
\caption{Comparing the quality of the solutions and run times for different values of $k$  for ten node instances.}
\label{Fig:k}
\end{figure}
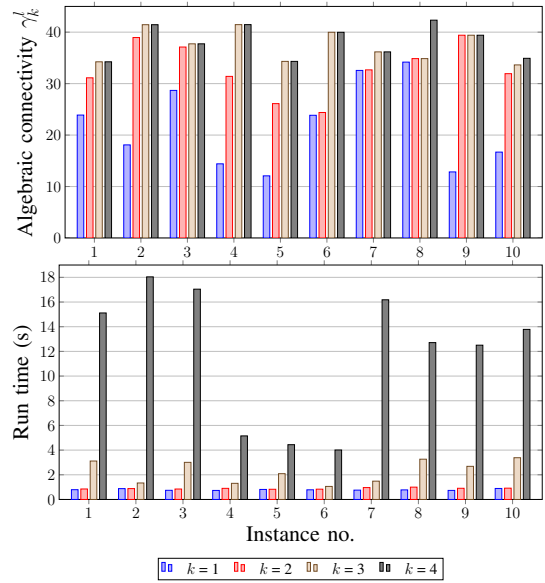

\subsection{Performance of the MCH}
\subsubsection{\underline{Solution quality of the MCH}}
Table \ref{Tab:MCH_quality} compares MCH with ${\mathcal F}_{k}^{l}$ for 10, 12, and 15-node instances, assessing heuristic quality and convergence speed. Gaps shown in Table \ref{Tab:MCH_quality} are evaluated as $(\frac{\gamma^{l}_{k} - \gamma_h}{\gamma^{l}_{k}} * 100)$. Here, $t_1$ and $t_2$ are the runtimes for MCH and $\mathcal{F}^l_k$, respectively. For all instance sizes, gaps consistently remain below \textbf{2\%}. With $(h_1, h_2)$ set to (5,5), MCH computes high-quality feasible solutions in approximately \textit{one second}.

\begin{table}[ht!]
\centering
    \caption{Average gaps and runtimes (for 50 instances) for proposed maximum cost heuristic (MCH) and the lower bounding formulation solutions for instances of varying sizes, with the corresponding $k$, $h_1$, and $h_2$ values.}
    \label{Tab:MCH_quality}
    \resizebox{0.8\columnwidth}{!}{%
    \begin{tabular}{p{0.75cm} p{0.75cm} p{0.75cm} p{0.75cm} p{1.5cm} p{1cm} p{1cm}}
    \Xhline{2\arrayrulewidth} \\ [-0.75em]
   $n$ &$k$ &$h_1$ & $h_2$ & gap (\%) &$t_1$ (s) &$t_2$ (s)\\ [0.5ex]
    \hline \\ [-0.5em]
    10 & 4 & 5 & 5 & 0.21 & \textbf{0.26} & 26.29 \\
    12 & 5 & 5 & 5 & 0.41 & \textbf{0.49} & 1912.89\\
    15 & 4 & 5 & 5 & 1.67 & \textbf{1.03} & 530.90\\ [1ex]
    \Xhline{2\arrayrulewidth}
    \end{tabular}%
    }
\end{table}
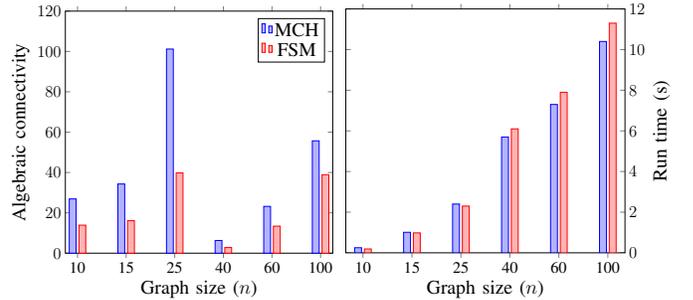
\begin{figure}[hbt!]
    \centering
    \resizebox{0.4925\columnwidth}{!}{%
\begin{tikzpicture}
\begin{axis}[
    x tick label style={/pgf/number format/1000 sep=},
    ybar=2.5pt,
    bar width=6pt,
    x=1.5cm,
    ymin=0,
    ymax=120,
    xticklabels = {, , 10, 15, 25, 40, 60, 100},
    enlarge x limits=0.05,
    legend pos=north east,
    legend style={font=\LARGE},
    xlabel= Graph size ($n$),
    ylabel={Algebraic connectivity},
    label style={font=\LARGE},
    tick label style={font=\Large}  
    ]
\addplot  coordinates {( 1,26.9) ( 2,34.3) ( 3,101.2) ( 4,6.3) ( 5,23.2) (6,55.7)};
\addplot coordinates {( 1,13.9) ( 2,16.1) ( 3,39.8) ( 4,2.8) ( 5,13.4) (6,38.9)};
\legend{MCH, FSM}
\end{axis}
\end{tikzpicture}%
}
\resizebox{0.4925\columnwidth}{!}{%
\begin{tikzpicture}
\begin{axis}[
    x tick label style={/pgf/number format/1000 sep=},
    ylabel near ticks, yticklabel pos=right,
    ybar=2.5pt,
    bar width=6pt,
    x=1.5cm,
    ymin=0,
    ymax=12,
    enlarge x limits=0.1,
    xticklabels = {, ,10, 15, 25, 40, 60, 100},
    enlarge x limits=0.07,
    xlabel= Graph size ($n$),
    ylabel={Run time (s)},
    label style={font=\LARGE},
    tick label style={font=\Large}  
    ]
\addplot  coordinates {( 1,0.25) ( 2,1.01) ( 3,2.4) ( 4,5.7) ( 5,7.3) (6,10.4)};
\addplot coordinates {( 1,0.19) ( 2,0.98) ( 3,2.3) ( 4,6.1) ( 5,7.9) (6,11.3)};

\end{axis}
\end{tikzpicture}%
}
\caption{\textcolor{black}{Comparing the solutions and runtimes of the proposed maximum cost heuristic (MCH) with parameters ($h_1, h_2$) set to (5, 3) against FSM \cite{son2010building} for networks up to one hundred nodes. The node degree upper bound is set to five.}}
\label{Fig:MCHvsFSM}
\end{figure}
\subsubsection{\underline{Comparison of MCH and FSM algorithms}}
\textcolor{black}{Figure \ref{Fig:MCHvsFSM} compares {\ac} values and runtimes of spanning trees generated by MCH and FSM algorithms, for the application discussed in Section \ref{subsec:Associated}. The comparison spans networks of various sizes, up to one hundred nodes, with an upper bound on the node degree set to $d = 5$.}

\textcolor{black}{The MCH algorithm consistently outperforms the FSM algorithm \cite{son2010building}, yielding spanning trees with higher {\ac} across all instances, with an average objective improvement of \textbf{17.32\%}. Additionally, MCH algorithm's runtimes are comparable to or better than the FSM algorithm for large-scale instances.}

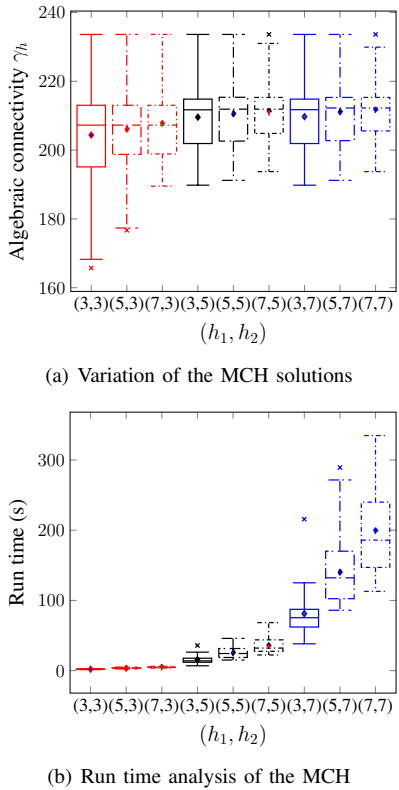
\begin{figure}[hbt!]
    \centering
    \subfigure[{Variation of the MCH solutions}]{
    \resizebox{0.6\columnwidth}{!}{%
    \begin{tikzpicture}
  \pgfplotsset{width=7.5cm, compat=1.5}
  \begin{axis}
    [
    boxplot/draw direction=y,
    xtick={1,2,3,4,5,6,7,8,9},
    scale only axis,
    xlabel=${(h_1,h_2)}$,
    ylabel=Algebraic connectivity $\gamma_h$,
    enlarge x limits=0.02,
    xticklabels={{(3,3)},{(5,3)},{(7,3)},{(3,5)},{(5,5)},{(7,5)},{(3,7)},{(5,7)},{(7,7)}},
    label style={font=\Large},
    tick label style={font=\large}  
    ]
    \addplot+[draw=red,thick,solid, mark = x,
    boxplot prepared={
      median=207.2673,
      average = 204.3963,
      upper quartile=213.0392,
      lower quartile=195.1228,
      upper whisker= 233.6343,
      lower whisker=168.2482
    },
    ] coordinates {(1,165.7182)};
    \addplot+[draw=red,thick,dash pattern={on 7pt off 2pt on 1pt off 3pt},mark = x,
    boxplot prepared={
      median=207.2673,
      average = 206.1439,
      upper quartile=213.0392,
      lower quartile=198.7631,
      upper whisker=233.6343,
      lower whisker=177.34895
    },
    ] coordinates {(1,176.704)};
    \addplot+[draw=red,thick,dash dot,mark = x,
    boxplot prepared={
      median=207.2673,
       average = 207.7205,
      upper quartile=213.0392,
      lower quartile=198.8882,
      upper whisker=233.6343,
      lower whisker= 189.5281
    },
    ] coordinates {};
    \addplot+[draw=black,thick,solid,mark = x,
    boxplot prepared={
      median=211.7223,
       average = 209.5786,
      upper quartile=214.826,
      lower quartile=201.9168,
      upper whisker=233.6343,
      lower whisker=189.8081
    },
    ] coordinates {};
    \addplot+[draw=black,thick,dash pattern={on 7pt off 2pt on 1pt off 3pt},mark = x,
    boxplot prepared={
      median=211.9033,
      average = 210.6189,
      upper quartile=215.36775,
      lower quartile=202.6205,
      upper whisker=233.6343,
      lower whisker=191.2118
    },
    ] coordinates {};
        \addplot+[draw=black,thick,dash dot,mark = x,
    boxplot prepared={
      median=211.9033,
       average = 211.2336,
      upper quartile=215.3194,
      lower quartile=204.8546,
      upper whisker=231.0166,
      lower whisker= 193.7585
    },
    ] coordinates {(1,233.6343)};
        \addplot+[draw=blue,thick,solid,mark = x,
    boxplot prepared={
      median=211.7598,
       average = 209.7275,
      upper quartile=214.826,
      lower quartile=201.9168,
      upper whisker=233.6343,
      lower whisker=189.8081
      },
    ] coordinates {};
        \addplot+[draw=blue,thick,dash pattern={on 7pt off 2pt on 1pt off 3pt},mark = x,
    boxplot prepared={
      median=212.253,
       average = 211.1981,
      upper quartile=215.3194,
      lower quartile=202.7518,
      upper whisker=233.6343,
      lower whisker=191.2118
    },
    ] coordinates {};
        \addplot+[draw=blue,thick,dash dot,mark = x,
    boxplot prepared={
      median=212.253,
      average = 211.6896,
      upper quartile=215.3194,
      lower quartile=205.5732,
      upper whisker=229.9387,
      lower whisker=193.7585
      },
    ] coordinates {(1,233.6343)};
  \end{axis}
\end{tikzpicture}
\label{Fig:H_25a}
}}
    \subfigure[{Run time analysis of the MCH}]{
    \resizebox{0.6\columnwidth}{!}{%
    \begin{tikzpicture}
  \pgfplotsset{width=7.5cm, compat=1.5}
  \begin{axis}
        [
    boxplot/draw direction=y,
    xtick={1,2,3,4,5,6,7,8,9},
    scale only axis,
    xlabel=${(h_1,h_2)}$,
    ylabel=Run time (s),
    enlarge x limits=0.02,
    xticklabels={{(3,3)},{(5,3)},{(7,3)},{(3,5)},{(5,5)},{(7,5)},{(3,7)},{(5,7)},{(7,7)}},
    label style={font=\Large},
    tick label style={font=\large} 
    ]
    \addplot+[draw=red,thick,solid,
    boxplot prepared={
      median=1.97,
      average = 2.05,
      upper quartile=2.09,
      lower quartile=1.86,
      upper whisker=2.7,
      lower whisker=1.36
    },
    ] coordinates {};
    \addplot+[draw=red,thick,dash pattern={on 7pt off 2pt on 1pt off 3pt},
    boxplot prepared={
      median=3.47,
      average =3.48,
      upper quartile=3.91,
      lower quartile=3.09,
      upper whisker=4.94,
      lower whisker=2.42
    },
    ] coordinates {};
    \addplot+[draw=red,thick,dash dot,
    boxplot prepared={
      median=4.79,
       average = 4.77,
      upper quartile=5.23,
      lower quartile=4.35,
      upper whisker=6.04,
      lower whisker=3.61
    },
    ] coordinates {};
    \addplot+[draw=black,thick,solid,mark = x,
    boxplot prepared={
      median=13.87,
       average = 15.61,
      upper quartile=17.68,
      lower quartile=11.93,
     upper whisker=26.3,
      lower whisker=6.92
    },
    ] coordinates {(1,35.81)};
    \addplot+[draw=black,thick,dash pattern={on 7pt off 2pt on 1pt off 3pt},
    boxplot prepared={
      median=24.38,
       average = 25.91,
      upper quartile=31.49,
      lower quartile=18.71,
      upper whisker=45.9,
      lower whisker=15.08
    },
    ] coordinates {};
        \addplot+[draw=black,thick,dash dot,
    boxplot prepared={
      median=32.1,
       average =35.52,
      upper quartile=43.82,
      lower quartile=27.46,
      upper whisker=68.36,
      lower whisker=22.27
    },
    ] coordinates {};
        \addplot+[draw=blue,thick,solid,mark = x,
    boxplot prepared={
      median=75.36,
       average =81.15,
      upper quartile=87.3,
      lower quartile=62.12,
      upper whisker=125.05,
      lower whisker=38.17
      },
    ] coordinates {(1,215.68)};
        \addplot+[draw=blue,thick,dash pattern={on 7pt off 2pt on 1pt off 3pt},mark = x,
    boxplot prepared={
      median=132.13,
       average =140.39,
      upper quartile=170.08,
      lower quartile=102.43,
      upper whisker=271.54,
      lower whisker=85.98
    },
    ] coordinates {(1,289.24)};
        \addplot+[draw=blue,thick,dash dot,
    boxplot prepared={
      median=185.82,
       average = 199.17,
      upper quartile=239.94,
      lower quartile=146.98,
      upper whisker=334.78,
      lower whisker=112.86
      },
    ] coordinates {};
  \end{axis}
\end{tikzpicture}
\label{Fig:H_25b}
}}

\caption{\textcolor{black}{Comparing the solutions and runtimes of the  maximum cost heuristic (MCH) for 25-node instances ($k=10$) with varying heuristic parameters. Line types represent variations in $h_1$, while colors denote changes in $h_2$.}}
\vspace{-0.3cm}
\label{Fig:H_25}
\end{figure}
\subsubsection{\underline{Variability of the MCH}}
\textcolor{black}{Due to the accelerated convergence of the MCH, high-quality feasible solutions are achieved for problem instances with sizes up to one hundred nodes. However, the solution quality and runtime are influenced by the selected heuristic parameters ($h_1, h_2$)}. 

Fig. \ref{Fig:H_25a} displays box plots of MCH solutions of twenty-five node instances for various heuristic parameter sets ($h_1, h_2$), illustrating {\ac} variations. Each box plot represents variation across fifty random instances. Higher values of $h_1$ or $h_2$ generally improve solution quality. Conversely, Fig. \ref{Fig:H_25b} shows an increasing runtime trend with higher heuristic parameter values. 
Opting for higher $h_1$ and reasonable $h_2$ values is preferable for obtaining high-quality solutions in less time. 
Implementing MCH utilizing Julia's parallel computing features can further reduce these runtimes.

\subsection{\textcolor{black}{Robustness of cooperative vehicle localization networks}}

\textcolor{black}{In Section \ref{subsec:coop_vehicle}, we introduced the problem related to cooperative vehicle localization under noisy measurements. To assess the robustness of different networks for this application, we compare the spectral norms of the state estimation error covariance matrix ($\mathbf{P}$) for ten vehicle instances. This analysis involves diverse communication networks, and the results are illustrated in Fig. \ref{Fig:robustness}. Specifically, our focus is on spanning tree networks among the vehicles, which include only $n-1$ communication links for $n$ vehicles. The box plot in Fig. \ref{Fig:robustness} shows the spectral norms of $\mathbf{P}$ for fifty random networks, encompassing star and chain networks, the maximum and minimum spanning tree networks, and the optimal network of $\mathcal{F}_1$ in \eqref{eq:F_1}. Throughout all instances in Fig. \ref{Fig:robustness}, the communication networks with the highest {\ac} consistently exhibit the smallest spectral norm of $\mathbf{P}$, indicated by red triangles. This observation signifies that the states of the vehicles are estimated most accurately with the communication network possessing the highest {\ac}, thereby emphasizing the robustness of this network.}

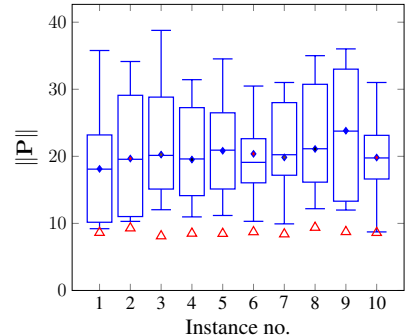
\begin{figure}[hbt!]
    \centering
    \resizebox{0.6\columnwidth}{!}{%
    \begin{tikzpicture}
  \pgfplotsset{width=7.5cm, compat=1.5}
  \begin{axis}
    [
    boxplot/draw direction=y,
    xtick={1,2,3,4,5,6,7,8,9,10},
    scale only axis,
    xlabel=Instance no.,
    ylabel=$\|\mathbf{P}\|$,
    ymin = 0,
    enlarge x limits=0.05,
    xticklabels={1, 2, 3, 4, 5, 6, 7, 8, 9, 10},
    label style={font=\Large},
    tick label style={font=\large} 
    ]
    \addplot+[draw=blue,thick,solid,
    boxplot prepared={
      median=18.09,
      average=18.13,
      upper quartile=23.19,
      lower quartile=10.17,
      upper whisker=35.78,
      lower whisker=9.21
    },
    ] coordinates {};
    \addplot+[draw=blue,thick,solid,
    boxplot prepared={
      median=19.56,
      average=19.67,
      upper quartile=29.1,
      lower quartile=11.0203,
      upper whisker=34.14,
      lower whisker=10.3
    },
    ] coordinates {};
    \addplot+[draw=blue,thick,solid,
    boxplot prepared={
      median=20.14,
      average=20.26,
      upper quartile=28.83,
      lower quartile=15.12,
      upper whisker=38.78,
      lower whisker=12.04
    },
    ] coordinates {};
    \addplot+[draw=blue,thick,solid,
    boxplot prepared={
      median=19.61,
      average=19.52,
      upper quartile=27.26,
      lower quartile=14.14,
      upper whisker=31.43,
      lower whisker=10.97
    },
    ] coordinates {};
    \addplot+[draw=blue,thick,solid,
    boxplot prepared={
      median=20.92,
      average=20.83,
      upper quartile=26.49,
      lower quartile=15.13,
      upper whisker=34.53,
      lower whisker=11.18
    },
    ] coordinates {};
        \addplot+[draw=blue,thick,solid,
    boxplot prepared={
      median=19.10,
      average=20.37,
      upper quartile=22.63,
      lower quartile=16.05,
      upper whisker=30.48,
      lower whisker= 10.31
    },
    ] coordinates {};
        \addplot+[draw=blue,thick,solid,
    boxplot prepared={
      median=20.24,
      average=19.84,
      upper quartile=28.01,
      lower quartile=17.18,
      upper whisker=31.01,
      lower whisker=9.93
    },
    ] coordinates {};
        \addplot+[draw=blue,thick,solid,
    boxplot prepared={
      median=21.13,
      average=21.10,
      upper quartile=30.74,
      lower quartile=16.15,
      upper whisker=35.01,
      lower whisker=12.19
    },
    ] coordinates {};
        \addplot+[draw=blue,thick,solid,
    boxplot prepared={
      median=23.76,
      average=23.82,
      upper quartile=33.00,
      lower quartile=13.31,
      upper whisker=36.02,
      lower whisker=11.99
    },
    ] coordinates {};
            \addplot+[draw=blue,thick,solid,
    boxplot prepared={
      median=19.76,
      average=19.82,
      upper quartile=23.13,
      lower quartile=16.62,
      upper whisker=31.01,
      lower whisker=8.73
    },
    ] coordinates {};
    \addplot[draw=red,thick,only marks,fill,mark=triangle,mark size=3.5pt] 
    coordinates {( 1,8.62) ( 2,9.29) ( 3,8.11) ( 4,8.50) ( 5,8.47) ( 6,8.72) ( 7,8.40) ( 8,9.38) ( 9,8.74) ( 10,8.61)};
  \end{axis}
\end{tikzpicture}
}
\caption{\textcolor{black}{Comparing the spectral norm of the state estimation error covariance matrix ($\mathbf{P}$) for 10 vehicle instances ($n=10$) for various communication networks. Red triangles denote spectral norms corresponding to optimal networks of $\mathcal{F}_1$ with maximum {\ac}.}}
\label{Fig:robustness}
\end{figure}

\section{Conclusions}
\label{Sec:conclusions}

This paper tackles the problem of maximizing the {\ac} for weighted networks, particularly in the context of cooperative vehicle localization under noisy measurements. A novel cutting plane-based upper bounding algorithm is introduced for this purpose, leveraging the principal minors characterization of positive semi-definite matrices.  The proposed algorithm demonstrates a notable improvement over existing methods by achieving tighter upper bounds with a reduced computational overhead compared to solving semi-definite programs with relaxed binary variables. Additionally, by integrating sparser principal sub-matrix cuts alongside denser cuts, the algorithm achieves faster runtimes. To address problems where obtaining optimal solutions is impractical, a degree-constrained Mixed-Integer Semi-Definite Programming (MISDP) formulation is presented to obtain lower bounds. \textcolor{black}{Furthermore, a maximum cost heuristic is proposed to quickly generate near-optimal solutions for larger networks, surpassing the performance of previously known fragment and selection-merging algorithms. Finally, a robustness comparison highlights the significance of selecting networks with higher {\ac} cooperative vehicle localization.}

\section*{Acknowledgements}
The authors gratefully acknowledge funding from Triad National Security LLC under the grant from the DOE NNSA (award no. 89233218CNA000001), titled ``Modeling and Discrete Optimization Algorithms for Robust Complex Networks'' and U.S. DOE's Laboratory Directed Research \& Development program under the project ``20230091ER: Learning to Accelerate Global Solutions for Non-convex Optimization''.

\bibliographystyle{IEEEtran}
\bibliography{References}

\section*{Appendix}
\begin{table}[ht]
\centering
\caption{Optimal ($\gamma^*$) and best found ($\gamma_{bfs}$) {\ac} values for MISDP in $\mathcal{F}_1$ across networks with up to fifteen nodes. They also serve as lower bounds for evaluating optimality gaps in section \ref{subsec:UB_perf}.}
\resizebox{0.9\columnwidth}{!}{%
\begin{tabular}{cccccc}
\Xhline{2\arrayrulewidth} \\ [-0.75em]
Instance & $\gamma^*$ & $\gamma^*$ & $\gamma^*$ & $\gamma_{bfs}$ & $\gamma_{bfs}$ \\  [0.5ex]
 & $n=8$ & $n=9$ &  $n=10$ & $n=12$ & $n=15$ \\ [0.5ex]
   \Xhline{2\arrayrulewidth} \\ [-0.75em]
    1	&	22.8042	& 28.2168	& 34.2371	&	54.0522	& 74.2785\\
    2	&	24.3207	& 26.3675	& 41.4488	&	53.2107	& 77.9973\\
    3	&	26.4111	& 29.8184 	& 37.7309	&	47.2228	&  80.0353\\
    4	&	28.6912	& 25.8427	& 41.4618	&	43.9330	& 89.7253\\
    5	&	22.5051	& 24.2756	& 34.3193	&	51.1286	& 77.2098\\
    6	&	25.2167	& 30.0202	& 39.9727	&	56.9622	& 64.1931\\
    7	&	22.8752	& 25.6410	& 36.1651	&	57.2901	& 80.7137\\
    8	&	28.4397	& 26.9705	& 42.3291	&	53.2338	& 75.7184\\
    9	&	26.7965	& 33.5068	& 39.4034	&	53.5628	& 85.7582\\
    10	&	27.4913	& 31.7445	& 34.9161	&	50.6987	& 77.7706\\[0.5ex]
   \Xhline{2\arrayrulewidth}
\end{tabular}
}
\label{tab:optimal_ac}
\end{table}

\end{document}